\documentclass[pra,twocolumn,preprintnumbers,amsmath,amssymb]{revtex4}
\input{Qcircuit}

\usepackage{ifpdf}
\ifpdf
\usepackage[pdftex]{graphicx,hyperref}
\else
\usepackage[dvips]{graphicx}
\fi
\usepackage{ifthen}

\newcommand{\ignore}[1]{}
\newcommand{\noises}{Refer to Tab. \ref{tab:errortypes} for descriptions of the types of noise.}

\newcommand{\figErrDetect}{
\begin{figure}[b]
  \caption{\label{fig:ErrDetect} General bit flip error detection process. If we are performing post-selection, we start all over if an error is detected.}
\vspace{-6mm}
  \begin{center}
    \[
      \Qcircuit @C=1em @R=.7em {
   	                &            & \ctrl{4} & \qw      & \qw      & \qw\\
\push{\rule{5.5em}{0em}}&	\lstick{\text{Data qubits}}    & \qw      & \ctrl{4} & \qw      & \qw\\
                      &            & \qw      & \qw      & \ctrl{4} & \qw\\
                      &            &          &          &          &  \\
   	                &            & \targ    & \qw      & \qw      & \meter\\
\push{\rule{5.5em}{0em}}&	\lstick{\text{Ancilla qubits}} & \qw      & \targ    & \qw      & \meter\\
                      &            & \qw      & \qw      & \targ    & \meter
      }
    \]
  \end{center}
\end{figure}
}

\newcommand{\figLogDetect}{
\begin{figure}[b]
  \caption{\label{fig:LogDet} Logical qubits under bit flip error detection}
\vspace{-6mm}
  \begin{center}
    \[
      \Qcircuit @C=1em @R=.7em {
\push{\rule{2em}{0em}} &   \lstick{\overline{\ket{\psi_D}}} & \ctrl{1} & \qw \\
\push{\rule{2em}{0em}} &   \lstick{\overline{\ket{\psi_A}}} & \targ & \meter
      }
    \]
  \end{center}
\end{figure}
}

\newcommand{\figPhysDetect}{
\begin{figure}[b]
  \caption{\label{fig:PhysDet} A pair of physical qubits under error detection. Pre-existing noise on the source qubit is given by $\mathcal{S}$, pre-existing noise on the destination qubit by $\mathcal{D}$. The noise from the {\tt CNOT} gate is $\mathcal{Q}$.}
\vspace{-6mm}
  \begin{center}
    \[
      \Qcircuit @C=1em @R=.7em {
\push{\rule{2em}{0em}} & \lstick{\ket{\psi}_S}       & \gate{\mathcal{S}} & \ctrl{2} & \multigate{2}{\mathcal{Q}} & \qw \\
                            &  \\
\push{\rule{2em}{0em}} & \lstick{\ket{\psi}_D}	    & \gate{\mathcal{D}} & \targ & \ghost{\mathcal{Q}} & \meter
      }
    \]
  \end{center}
\end{figure}
}

\newcommand{\figTeleportation}{
\begin{figure}[b]
  \caption{\label{fig:Teleportation} Teleportation circuit }
\vspace{-6mm}
  \begin{center}
    \[
      \Qcircuit @C=1em @R=.7em {
\push{\rule{2em}{0em}} &   \lstick{\ket{\psi}} & \ctrl{1} & \qw     & \measure{X} \cwx[3]\\
	                    && \targ    & \meter \cwx[2]  \\
\push{\rule{6em}{0em}} &   \lstick{\frac{1}{\sqrt{2}}(\ket{00}+\ket{11})} \\
                       && \qw      & \gate{X}  & \gate{Z} & \qw & \rstick{\ket{\psi}}
      }
    \]
  \end{center}
\end{figure}
}

\newcommand{\figTeleportationSplit}{
\begin{figure}[b]
  \caption{\label{fig:TeleportationSplit} Teleportation splitting circuit.  }
\vspace{-6mm}
  \begin{center}
    \[
      \Qcircuit @C=1em @R=.7em {
\push{\rule{5em}{0em}} &   \lstick{c_0 \ket{0} + c_1 \ket{1}} & \ctrl{1} & \qw     & \measure{X} \cwx[2]\\
	                    && \targ    & \meter \cwx[1]  \\
\push{\rule{7em}{0em}} &   \lstick{\frac{1}{\sqrt{2}}(\ket{000}+\ket{111})} & \qw & \gate{X} \cwx[2] &  \gate{Z} & \qw \\
                       &&&&& \rstick{c_0 \ket{00} + c_1 \ket{11}}\\
                       && \qw      & \gate{X}  & \qw & \qw  \\
      }
    \]
  \end{center}
\end{figure}
}

\newcommand{\figTeleportationMerge}{
\begin{figure}[b]
  \caption{\label{fig:TeleportationMerge} Teleportation merging circuit.  }
\vspace{-6mm}
  \begin{center}
    \[
      \Qcircuit @C=1em @R=.7em {
\push{\rule{2em}{0em}} &	\lstick{\ket{c}}          & \ctrl{2} & \qw      & \qw            &  \qw             & \measure{X} \cwx[4]\\
\push{\rule{2em}{0em}} &	\lstick{\ket{d}}          & \qw      & \ctrl{2} & \qw            &  \qw             & \qw & \measure{X} \cwx[3] \\
\push{\rule{2em}{0em}} &	\lstick{\ket{a}}          & \targ    & \qw      & \meter \cwx[2]\\
\push{\rule{2em}{0em}} &	\lstick{\ket{b}}          & \qw      & \targ    & \qw            &  \meter \cwx[1] \\
\push{\rule{3em}{0em}} &	\lstick{\ket{a \oplus b}} & \qw      & \qw      & \gate{X}       &  \gate{X}        & \gate{Z} & \gate{Z} & \qw & \rstick{\ket{c \oplus d}} \\
      }
    \]
  \end{center}
\end{figure}
}

\newcommand{\figCNOTTeleportation}{
\begin{figure}[b]
  \caption{\label{fig:CNOTTeleportation} {\tt CNOT} gate via Teleportation. $\ket{\psi_3}$ is $\ket{\psi_1}$ {\tt CNOT} $\ket{\psi_2}$. }
\vspace{-6mm}
  \begin{center}
    \[
      \Qcircuit @C=1em @R=.7em {
\push{\rule{2em}{0em}} &   \lstick{\ket{\psi_1}} &  \ctrl{2} & \qw      & \qw   & \qw & \measure{X} \cwx[4]  \\
\push{\rule{2em}{0em}} &	\lstick{\ket{\psi_2}} &  \qw      & \ctrl{2} & \qw   & \qw & \qw & \measure{X}  \cwx[4]\\
\push{\rule{2em}{0em}} &   \lstick{\ket{a}}      &  \targ    & \qw      & \meter \cwx[2] \\
\push{\rule{2em}{0em}} &   \lstick{\ket{b}}      & \qw       & \targ    & \qw   &\meter \cwx[2]\\
\push{\rule{2em}{0em}} &   \lstick{\ket{a}}      & \qw       & \qw      &\gate{X} \cwx[1] & \qw      & \gate{Z} & \qw  & \qw &  \rstick{ \ket{\psi_1}} \\
\push{\rule{3em}{0em}} &   \lstick{\ket{a \oplus b}} & \qw       & \qw      &\gate{X} & \gate{X} & \qw & \gate{Z}  & \qw& \rstick{\ket{\psi_3}}
      }
    \]
  \end{center}
\end{figure}
}

\newcommand{\figNNCNOT}{
\begin{figure}[b]
  \caption{\label{fig:encode} Creating encoded $\overline{\ket{0}}$ for the $[[7,1,3]]$ CSS code. This results in the state $\overline{\ket{0}} = \frac{1}{\sqrt{8}} (\ket{0000000} + \ket{0001111} + \ket{0110011} + \ket{0111100} + \ket{1010101} + \ket{1011010} + \ket{1100110} + \ket{1101001})$.}
\vspace{-6mm}
  \begin{center}
    \[
      \Qcircuit @C=1em @R=.7em {
	\lstick{\ket{+}} & \ctrl{6} & \qw      & \qw      \\
        \lstick{\ket{+}} & \qw      & \ctrl{5} & \qw      \\
        \lstick{\ket{0}} & \targ    & \targ    & \qw      \\
        \lstick{\ket{+}} & \qw      & \qw      & \ctrl{3} & \rstick{\overline{\ket{0}}} \\
        \lstick{\ket{0}} & \targ    & \qw      & \targ    \\
        \lstick{\ket{0}} & \qw      & \targ    & \targ    \\
        \lstick{\ket{0}} & \targ    & \targ    & \targ    
      }
    \]
  \end{center}
\end{figure}
}

\newcommand\braket[2]{{ \langle {#1} | {#2} \rangle}}

\newtheorem{theorem}{Theorem}
\newtheorem{proposition}[theorem]{Proposition}

\newtheorem{corollary}[theorem]{Corollary}

\newcommand{\qedsymb}{\hfill{\rule{2mm}{2mm}}}
\newenvironment{proof}[1][]{\begin{trivlist}
\item[\hspace{\labelsep}{\bf\noindent Proof#1:\/}] }{\qedsymb\end{trivlist}}

\begin{document}
\title{An upper bound on quantum fault tolerant thresholds}
\author{Jesse Fern}

\begin{abstract}
In this paper we calculate upper bounds on fault tolerance, without restrictions on the overhead involved. Optimally adaptive recovery operators are used, and the Shannon entropy is used to estimate the thresholds. By allowing for unrealistically high levels of overhead, we find a quantum fault tolerant threshold of $6.88\%$ for the depolarizing noise used by Knill \cite{Knill}, which compares well to "above $3\%$" evidenced by Knill. We conjecture that the optimal threshold is $6.90\%$, based upon the hashing rate. We also perform threshold calculations for types of noise other than that discussed by Knill. 
\end{abstract}

\affiliation{Berkeley Center for Quantum Information and Computation, Department of Mathematics, University of California, Berkeley, California, 94720
}
\email{jesse@math.berkeley.edu}

\maketitle

\section{Introduction}

In this paper, we calculate the quantum fault tolerant thresholds for the $[[7,1,3]]$ and $[[23,1,7]]$ CSS codes under massive amounts of post-selection. Post-selection \cite{Knill3} is the process of performing error detection, and then rejecting the result if an error is detected. We optimize the recovery operators under concatenation \cite{MoreChans} to find the optimal thresholds. Here there are no restrictions on the quantity of overhead involved. Optimizations of the error threshold with respect to the overhead will be considered in future work.

Suppose that, given a logical data qubit $\ket{\psi_D}$ encoded into some physical data qubits and a logical ancilla qubit $\ket{\psi_A}$ encoded into some physical ancilla qubits, we wish to detect bit flip errors in the data qubits. We apply {\tt CNOT} gates transversally from the data qubits to the ancilla qubits, and then measure the ancilla qubits, as shown in Fig. \ref{fig:ErrDetect}. The fact that the syndrome measurement may be faulty is included in our analysis.  For post-selection, we start all over if an error is detected. 
\figErrDetect

This process works well for CSS codes and bit flip qubits. {\tt CNOT} is encoded as itself  acting transversally on each pair of qubits. In terms of the logical qubits, the result will resemble Fig. \ref{fig:LogDet}. For an arbitrary logical data qubit $\overline{\ket{\psi_D}}$, we use $\overline{\ket{\psi_A}} = \overline{\ket{+}}$, where $\ket{+} = \frac{1}{\sqrt{2}} (\ket{0} + \ket{1})$ for the logical ancilla qubit. If the logical source qubit is an encoded Pauli eigenstate of $\overline{Z}$, that is $\overline{\ket{0}}$ or $\overline{\ket{1}}$, then it is more efficient for the logical destination ancilla qubit to also be an encoded Pauli eigenstate of $\overline{Z}$. 

\figLogDetect

We assume the ideal case that the noise on each qubit is independent, and that any bit flip error that was sent to or was already existing on the destination (ancilla) qubits can be detected.  We model this by looking at each physical data qubit and physical ancilla qubit pair separately. However, it is not exactly true that all of these errors can be detected; if $w$ is the lowest weight (that is, the number of qubits for which the Pauli matrix is not the identity on that qubit) of a stabilizer element (other than the identity) for the code, then there are some cases where a total of $w$ errors in the source or destination qubits can result in an undetected stabilizer error in the destination qubits.

We run error detection separately for bit flip errors ($X$ or $Y$) and phase flip errors ($Z$ or $Y$). One way to do this is to apply Hadamard gates in between the error detection steps to change bit flip errors into phase flip errors, and vice versa. In this paper, we apply post-selection to reduce the rate of errors in some ancilla qubits.

In Sec. \ref{sec:basic}, the noise resulting from a {\tt CNOT} gate followed by measurement of the destination qubit is found. The same results are derived using the channel maps \cite{FKSS,MoreChans} of quantum codes in Appendix \ref{sec:codemethod}. Both of these methods assume the same noise on each qubit.

In Sec. \ref{sec:special}, we apply these results to a code encoded in some specific ancilla qubits that allow for an encoded {\tt CNOT} with error correction using just one {\tt CNOT} per logical data qubit. We again assume the same noise on each qubit.

In Sec. \ref{sec:noisemodels}, we consider realistic types of noise, and calculate the relevant quantum fault tolerant thresholds. In Sec. \ref{sec:entropy}, we examine how the fact that the noise is not the same on each qubit affects the threshold. We use the Shannon information entropy \cite{Shannon} as in \cite{MoreChans} to calculate the actual thresholds.

 In Sec. \ref{sec:converge}, we look at how well the noise is corrected when we are below the fault tolerant threshold.


\section{Basic calculations}
\label{sec:basic}
\subsection{Channel Notation}
A qubit has two orthogonal states, $\ket{0}$ and $\ket{1}$. If there is a set of states $\ket{\psi_i}$ with corresponding probabilities of occurring $p_i$, the resulting mixed state is represented by the density matrix
\begin{equation*}
\rho = \sum_i p_i \ket{\psi_i} \bra{\psi_i}.
\end{equation*}
We can write this $2$ by $2$ matrix in terms of the $4$ one qubit Pauli matrices $I$, $X$, $Y$ and $Z$ as
\begin{equation*}
\rho = \frac{1}{2}(I + c_X X + c_Y Y + c_Z Z).
\end{equation*}
Alternately, one can write this density matrix in vector form as
\begin{equation}
\label{eq:vectordensity}
v(\rho) = \begin{bmatrix}
1 \cr
c_X \cr
c_Y \cr
c_Z \cr
\end{bmatrix},
\end{equation}
which we shall call a density vector. Now we consider the map from density matrices to density matrices. These maps are of the form
\begin{equation*}
\rho \rightarrow \sum_j A_j \rho A_j^\dagger, \quad
\mathrm{where} \quad \sum_j A_j^\dagger A_j = I.
\end{equation*}
They can be written as the channel superoperator matrix that acts on the density vectors as 
\begin{equation*}
\label{noise}
\mathcal{N}^{(1)} \equiv
\begin{bmatrix}
1 & 0 & 0 & 0 \cr
N_{XI} & N_{XX} & N_{XY} & N_{XZ} \cr
N_{YI} & N_{YX} & N_{YY} & N_{YZ} \cr
N_{ZI} & N_{ZX} & N_{ZY} & N_{ZZ}
\end{bmatrix},
\end{equation*}
where the $N_{\sigma \sigma'}$ are all real.

A unitary matrix $U$ acts upon a density state as $\rho \rightarrow U \rho U^{\dagger}$. We represent it in the channel superoperator form as a superoperator  $\mathcal{O}(U)$. We represent these in the Pauli basis. A Pauli matrix $\sigma$ has the diagonal superoperator $\mathcal{O}(\sigma)$, where each of the coefficients $\mathcal{O}(\sigma)_{\sigma' \sigma'}$ is $1$ if $\sigma$ and $\sigma'$ commute, and $-1$ if they anti-commute. For the $1$ qubit Pauli matrices, the superoperators are diagonal and are written as
\begin{align}
\label{eq:pauliprob}
\begin{split}
\mathcal{O}(I) = [1,1,1,1] &\quad \mathcal{O}(X) = [1,1,-1,-1] \\
\mathcal{O}(Y) = [1,-1,1,-1] &\quad \mathcal{O}(Z) = [1,-1,-1,1]
\end{split},
\end{align}
where we record the diagonal elements in vector form. We can write a diagonal channel in terms of the probabilities $p_\sigma$ of having the Pauli error $\sigma$ as
\begin{equation*}
[1, N_{XX}, N_{YY}, N_{ZZ}] =  \sum_{\sigma} p_\sigma \mathcal{O}(\sigma)= [1,x,y,z] ,
\end{equation*}
which, using Eq. \ref{eq:pauliprob}, gives
\begin{align}
\label{eq:chanpauli}
x &= 1 - 2 (p_Y + p_Z) \nonumber \\
y &= 1 - 2 (p_X + p_Z) \\
z &= 1 - 2 (p_X + p_Y). \nonumber
\end{align}


\subsection{Measurement channels}

\begin{proposition}
If we have a pure state $\nu = \ket{\psi} \bra{\psi}$ and a state $\rho$, and we measure the state $\rho$ in an orthonormal basis which includes $\ket{\psi}$, then the probability of measuring $\ket{\psi}$ is the fidelity, which is half of the inner product of the vectorized density matrices (Eq. \ref{eq:vectordensity}): 
\begin{equation*}
F(\rho, \nu) = \frac{1}{2} \braket{v(\rho)}{v(\nu)}.
\end{equation*}
\end{proposition}
\begin{proof}
$\rho$ can be written as a linear combination of pure states as $\rho = \sum_i d_i \ket{\gamma_i} \bra{\gamma_i}$, and so the probability of measuring $\ket{\psi}$ is  
\begin{equation*}
F = \sum_i d_i \braket{\gamma_i}{\psi}^2 = \sum_i d_i \braket{\psi}{\gamma_i} \braket{\gamma_i}{\psi} = \bra{\psi} \rho \ket{\psi}.
\end{equation*}
Now, $\text{tr } \rho \nu = \sum_i \bra{a_i} \rho \nu \ket{a_i}$ for any basis $\ket{a_i}$, so
\begin{equation*}
\text{tr } \rho \nu = \text{tr } \rho \ket{\psi} \bra{\psi} = \bra{\psi} \rho \ket{\psi} \braket{\psi}{\psi} = \bra{\psi} \rho \ket{\psi} = F.
\end{equation*}
If
\begin{equation*}
v(\rho) = \begin{bmatrix}
1 \cr
\rho_X \cr
\rho_Y \cr
\rho_Z \cr
\end{bmatrix}, v(\nu) = \begin{bmatrix}
1 \cr
\nu_X \cr
\nu_Y \cr
\nu_Z \cr
\end{bmatrix},
\end{equation*}
then
\begin{equation*}
\text{tr } \rho \nu = \frac{1}{2} \sum_{\sigma} \rho_{\sigma} \nu_{\sigma} = \frac{1}{2} \braket{v(\rho)} {v(\nu)}.
\end{equation*}
\end{proof}

From this it follows that:
\begin{corollary}
\label{cor:measure}
If a noise $\mathcal{N}$ is applied to the state $\rho$, the probability of measuring the pure state $\nu$ is
\begin{equation*}
\frac{1}{2} \bra{v(\nu)} \mathcal{N} \ket{ v(\rho)}.
\end{equation*}
\end{corollary}
 For example, if they are both the pure state $\ket{0}$, $\rho = \nu = \frac{1}{2}(I+Z)$, and so
\begin{equation*}
\frac{1}{2} \begin{bmatrix}
1 & 0 & 0 & 1 \cr
\end{bmatrix} \mathcal{N} \begin{bmatrix}
1 \cr
0 \cr
0 \cr
1 \cr
\end{bmatrix} = \frac{N_{II}+N_{IZ}+N_{ZI}+N_{ZZ}}{2}.
\end{equation*}

\begin{theorem}
\label{thm:correct}
If a noise $\mathcal{N}$ is applied to a qubit, and then the qubit is measured in the Pauli operator $\sigma$ basis, which yields one of the Pauli eigenstates $\ket{\sigma_+}$ or $\ket{\sigma_-}$, the correct result is obtained by the measurement with probability
\begin{equation*}
m_c = \frac{1 + N_{\sigma \sigma}}{2}.
\end{equation*}
\end{theorem}
\begin{proof}
We have initial state $\rho_i$ and measured state $\rho_m$, which are eigenstates of $\sigma$, with associated eigenvalues $\lambda_i$ and $\lambda_m$. Then, from Cor. \ref{cor:measure}, after a noise $\mathcal{N}$ occurs, the probability of measuring the initial state $\rho_i$ as the measured state $\rho_m$ is 
\begin{align*}
\frac{1}{2} \bra{v(\rho_m)} \mathcal{N} \ket{v(\rho_i)} = \frac{1}{2}
\begin{bmatrix}
1 & \lambda_m  \cr
\end{bmatrix} \begin{bmatrix}
N_{II}       & N_{I \sigma} \cr
N_{\sigma I} & N_{\sigma \sigma} \cr
\end{bmatrix} \begin{bmatrix}
1\cr
\lambda_i  \cr
\end{bmatrix} \\
= \frac{N_{II} +\lambda_1 N_{I \sigma} +\lambda_2 N_{\sigma I}+ \lambda_1 \lambda_2 N_{\sigma \sigma}}{2}.
\end{align*}
Tab. \ref{tab:measureprob} gives the probabilities for the $4$ cases in which the initial eigenvalue $\lambda_i$ is either $-1$ or $+1$ and the measured eigenvalue $\lambda_m$ is either $-1$ or $+1$.

\ignore{
Given some Pauli operator $\sigma$, the probability of measuring a state that was initially 
Given some Pauli operator $\sigma$, and $\rho$  $\lambda$ eigenstate 
Suppose that for some Pauli operator $\sigma$, $\rho$ is the $\lambda_1$ eigenstate, and $\rho_M$ is the $\lambda_2$ eigenstate. Then the probability of measuring the 

Suppose that $\rho$ and $\rho_M$ are both eigenstates of the same Pauli operator $\sigma$. Then, the probability of measuring the eigenvalue $\pm_2 1$ is
\begin{align*}
\frac{1}{2}
\begin{bmatrix}
1 & \pm_2 1 \cr
\end{bmatrix} \begin{bmatrix}
N_{II}       & N_{I \sigma} \cr
N_{\sigma I} & N_{\sigma \sigma} \cr
\end{bmatrix} \begin{bmatrix}
1\cr
\pm_1 1 \cr
\end{bmatrix} \\
= \frac{N_{II} \pm_1 N_{I \sigma} \pm_2 N_{\sigma I} \pm_1 \pm_2 N_{\sigma \sigma}}{2}
\end{align*}
when the initial eigenvalue was $\pm_1 1$, and the $\pm_1$ and $\pm_2$ act independently of each other. The probabilities of each of these $4$ cases of having the initial state and the measured state in the $-1$ or $+1$ eigenstates is given in Tab. \ref{tab:measureprob}.
}

\begin{table}
\caption{\label{tab:measureprob} Measurement probabilities}
\begin{ruledtabular}
\begin{tabular}{cccc}
Measured state  &  Initial state $\lambda_i=1$ & Initial state $\lambda_i=-1$ \\
\hline
$\lambda_m=1$  &  $\frac{N_{II}+N_{I \sigma}+N_{\sigma I}+N_{\sigma \sigma }}{2}$ & $\frac{N_{II}-N_{I \sigma}+N_{\sigma I}-N_{\sigma \sigma}}{2}$ \cr
$\lambda_m=-1$  &  $\frac{N_{II}+N_{I \sigma}-N_{\sigma I}-N_{\sigma \sigma}}{2}$ & $\frac{N_{II}-N_{I \sigma }-N_{\sigma I}+N_{\sigma \sigma}}{2}$ \cr
\end{tabular}
\end{ruledtabular}
\end{table}
The total probability of a wrong measurement averages to $\frac{1-N_{\sigma \sigma}}{2}$, and the total probability of a correct measurement averages to $\frac{1+N_{\sigma \sigma}}{2}$.
\end{proof}

\subsection{Noise from applying {\tt CNOT}}
As shown in Fig. \ref{fig:PhysDet}, we have a source qubit $\ket{\psi}_S$ with pre-existing local noise $\mathcal{S}$, and a  destination qubit $\ket{\psi}_D$ with pre-existing local noise $\mathcal{D}$. When we apply {\tt CNOT}, the two qubit noise $\mathcal{Q}$ occurs on these qubits. The total superoperator for this process is found by multiplying the initial noise $\mathcal{S} \otimes \mathcal{D}$ first by the superoperator for {\tt CNOT} and then by the new noise $\mathcal{Q}$ from {\tt CNOT}. This results in 
\begin{equation*}
\mathcal{T} = Q \circ \mathcal{O}({\tt CNOT}) \circ (S \otimes D),
\end{equation*}
where $\circ$ represents the matrix multiplication of superoperators. Since we wanted  only {\tt CNOT} to be applied, the total noise can be represented in terms of a noise superoperator $\mathcal{N}$ as 
\begin{equation}
\mathcal{T} = \mathcal{N} \circ \mathcal{O}({\tt CNOT})  = Q \circ \mathcal{O}({\tt CNOT}) \circ (S \otimes D) \circ \mathcal{O}({\tt CNOT}).
\end{equation}

\figPhysDetect

\subsection{Diagonal noise}
In the rest of this paper, we assume that the noise is diagonal. We utilize the shorthand $\mathcal{N}_{\sigma} \equiv \mathcal{N}_{\sigma \sigma}$. This noise maps a density matrix $I + \sigma \rightarrow I + N_{\sigma} \sigma$. The noise on the source qubits has diagonal noise $S = [1, S_X, S_Y, S_Z]$, and the noise on the destination qubits has diagonal noise $D=[1,D_X,D_Y,D_Z]$.

If the off-diagonal terms are small, they can be neglected anyway. If they are $o(\epsilon)$, the threshold is only affected by $o(e^m)$, where $m$ is the minimum weight of the non-identity elements that stabilize the code \cite{FKSS}. For the $[[7,1,3]]$ CSS code, $m=4$. For the $[[23,1,7]]$ CSS code, $m=8$.

To determine how the {\tt CNOT} gate effects the errors, we look at the map of the Pauli matrices under conjugation by {\tt CNOT}. This gives the map
\begin{align*}
(II), (IX), IY \leftrightarrow ZY, IZ \leftrightarrow ZZ, XI \leftrightarrow XX\\
 XY \leftrightarrow YZ, XZ \leftrightarrow YY, YI \leftrightarrow YX, (ZI), (ZX).
\end{align*}
The diagonal noise $\mathcal{R} = \mathcal{O}({\tt CNOT}) \circ (S \otimes D) \circ \mathcal{O}({\tt CNOT})$ is given by Table \ref{tab:R}. From this, the resulting diagonal noise $\mathcal{N} = \mathcal{Q} \circ \mathcal{R}$ is given in Table \ref{tab:diagonal}.

\begin{table}
\caption{\label{tab:R} The noise $\mathcal{R} = \mathcal{O}({\tt CNOT}) \circ (S \otimes D) \circ \mathcal{O}({\tt CNOT})$ resulting from $\mathcal{S} \otimes \mathcal{D}$ conjugated with {\tt CNOT}.}
\begin{ruledtabular}
\begin{tabular}{ccccc}
$\sigma$ & $\mathcal{R}_{\sigma I}$ & $\mathcal{R}_{\sigma X}$ & $\mathcal{R}_{\sigma Y}$ & $\mathcal{R}_{\sigma Z}$ \cr
\hline
$I$ & $S_I D_I$ & $S_I D_X$ & $S_Z D_Y$ & $S_Z D_Z$ \cr
$X$ & $S_X D_X$ & $S_X D_I$ & $S_Y D_Z$ & $S_Y D_Y$ \cr
$Y$ & $S_Y D_X$ & $S_Y D_I$ & $S_X D_Z$ & $S_X D_Y$ \cr
$Z$ & $S_Z D_I$ & $S_Z D_X$ & $S_I D_Y$ & $S_I D_Z$ \cr
\end{tabular}
\end{ruledtabular}
\end{table}

\begin{table}
\caption{\label{tab:diagonal} The total noise $\mathcal{N}$ resulting from the error detection process.}
\begin{ruledtabular}
\begin{tabular}{ccccc}
$\sigma$ & $\mathcal{N}_{\sigma I}$ & $\mathcal{N}_{\sigma X}$ & $\mathcal{N}_{\sigma Y}$ & $\mathcal{N}_{\sigma Z}$ \cr
\hline
$I$ & $S_I D_I Q_{II}$ & $S_I D_X Q_{IX}$ & $S_Z D_Y Q_{IY}$ & $S_Z D_Z Q_{IZ}$ \cr
$X$ & $S_X D_X Q_{XI}$ & $S_X D_I Q_{XX}$ & $S_Y D_Z Q_{XY}$ & $S_Y D_Y Q_{XZ}$ \cr
$Y$ & $S_Y D_X Q_{YI}$ & $S_Y D_I Q_{YX}$ & $S_X D_Z Q_{YY}$ & $S_X D_Y Q_{YZ}$ \cr
$Z$ & $S_Z D_I Q_{ZI}$ & $S_Z D_X Q_{ZX}$ & $S_I D_Y Q_{ZY}$ & $S_I D_Z Q_{ZZ}$ \cr
\end{tabular}
\end{ruledtabular}
\end{table}

\subsubsection{After $Z$ measurement on destination qubit}
\label{ancillaprepare}
Now we assume that the destination qubit is set up so that we can detect a $Z$ error upon measuring it in the $Z$ basis. From Theorem \ref{thm:correct}, we have two types of syndromes, the "correct" one which corresponds to no error detected, and the "incorrect" one which corresponds to an error detected. These have probabilities $\frac{1 + Z}{2}$ and $\frac{1-Z}{2}$. By tracing out $N$ on the destination qubit, these result in 
\begin{align}
\label{eq:Ntraceout}
\mathcal{G}^{II}_{\sigma} = \frac{N_{\sigma I} + N_{\sigma Z}}{2} &&
\mathcal{G}^{IX}_{\sigma} = \frac{N_{\sigma I} - N_{\sigma Z}}{2}.
\end{align}
This gives
\begin{align}
\label{eq:measure}
\mathcal{G}^{II} = \frac{1}{2}(A+B) \quad \mathrm{and} \quad  \mathcal{G}^{IX} = \frac{1}{2} (A-B),
\end{align}
where

\begin{align*}
A = [1, S_X D_X Q_{XI}, S_Y D_X Q_{YI}, S_Z Q_{ZI}] \\
B = [S_Z D_Z Q_{IZ}, S_Y D_Y Q_{XZ}, S_X D_Y Q_{YZ}, D_Z Q_{ZZ}]
\end{align*}
are the 1st and 4th columns of Tab. \ref{tab:diagonal}, respectively.

\subsubsection{After $X$ measurement on source qubit and $Z$ measurement on destination qubit}
Like before, we trace out, this time with $\frac{1 \pm X}{2}$. Tracing out Eq. \ref{eq:measure}, we get probabilities of Pauli errors
\ignore{
\begin{align*}
p_I = \frac{1}{4} ( 1 + S_X D_X Q_{XI} + S_Z D_Z Q_{IZ} + S_Y D_Y Q_{XZ})\\
p_X = \frac{1}{4} ( 1 + S_X D_X Q_{XI} - S_Z D_Z Q_{IZ} - S_Y D_Y Q_{XZ})\\
p_Y = \frac{1}{4} ( 1 - S_X D_X Q_{XI} - S_Z D_Z Q_{IZ} + S_Y D_Y Q_{XZ})\\
p_Z = \frac{1}{4} ( 1 - S_X D_X Q_{XI} + S_Z D_Z Q_{IZ} - S_Y D_Y Q_{XZ}).
\end{align*}
}
\begin{equation}
\label{eq:Teleporterror1}
\begin{bmatrix}
p_I \cr
p_X \cr
p_Y \cr
p_Z \cr
\end{bmatrix} = \frac{1}{4} \begin{bmatrix}
1 & 1  &  1 & 1 \cr
1 & 1  & -1 & -1 \cr
1 & -1 &  1 & -1 \cr
1 & -1 & -1 & 1 \cr
\end{bmatrix} \begin{bmatrix}
1 \cr
S_X D_X Q_{XI} \cr
S_Y D_Y Q_{XZ} \cr
S_Z D_Z Q_{IZ} \cr
\end{bmatrix}.
\end{equation}
Note that only bit flip errors on the destination qubit and phase flip errors on the source qubit contribute to these errors. 

\subsection{Measurement errors}
\label{measurementerrors}
Suppose that we have probability $p_m$ of an $X$ measurement error. This is equivalent to probability $p_X=p_m$ of $X$ error and probability $p_I=1-p_m$ of an $I$ error (no error) right before measurement. From Eq. \ref{eq:pauliprob}, this has the channel 
\begin{equation*}
[1,x,y,z] = (1-p_m)\mathcal{O}(I)+p_m \mathcal{O}(X) = [1,1,m,m],
\end{equation*}
where 
\begin{equation}
\label{eq:m}
m = 1 - 2p_m. 
\end{equation}
Similarly, $Z$ measurement errors produce the channel $[1,x,y,z] = [1,m,m,1]$.

This changes the resulting noise after measurement from Eq. \ref{eq:measure} to
\begin{align*}
\mathcal{G}^{II} = \frac{1}{2}(A+mB) \quad \mathrm{and} \quad  \mathcal{G}^{IX} = \frac{1}{2} (A-mB).
\end{align*}
Then Eq. \ref{eq:Teleporterror1} becomes
\begin{equation*}
\begin{bmatrix}
p_I \cr
p_X \cr
p_Y \cr
p_Z \cr
\end{bmatrix} = \frac{1}{4} \begin{bmatrix}
1 & 1  &  1 & 1 \cr
1 & 1  & -1 & -1 \cr
1 & -1 &  1 & -1 \cr
1 & -1 & -1 & 1 \cr
\end{bmatrix} \begin{bmatrix}
1 \cr
m S_X D_X Q_{XI} \cr
m^2 S_Y D_Y Q_{XZ} \cr
m S_Z D_Z Q_{IZ} \cr
\end{bmatrix}.
\end{equation*}

\section{Special ancilla qubits}
\label{sec:special}
Suppose that we have a high distance (minimum weight of an undetected error) code. We have found ancilla qubit encodings that allow for a higher quantum fault tolerant threshold than other methods \cite{Knill}.  These ancilla qubits are logical qubits encoded in a high distance quantum code. We then use quantum teleportation to bring in logical data qubits.  

Once we have sufficiently reduced the rates of errors in these qubits via post-selection, we use quantum teleportation \cite{BBCJPW} to replace the logical ancilla qubits with logical data qubits, as is also done in the work of Knill \cite{Knill3} and Reichardt \cite{Reichardt}. The simplest teleportation is to use a Bell pair (two qubit cat state)
\begin{equation*}
\frac{1}{\sqrt{2}} (\ket{00} + \ket{11}) = \frac{1}{\sqrt{2}} \sum_{a=0}^1 \ket{aa},
\end{equation*}
and apply the standard teleportation gate given in Fig. \ref{fig:Teleportation}. If $\ket{b}$ represents the standard basis elements of $\ket{\psi}$, then we start with $\ket{b} \ket{a} \ket{a}$. After {\tt CNOT} is applied from the first to the second qubit, it becomes $\ket{b} \ket{a \oplus b} \ket{a}$. Measuring the second qubit results in $\ket{b} \ket{b}$ if $0$ is measured, and $\ket{b} \oplus \ket{b \oplus 1}$ if $1$ is measured. Applying $X$ to the last qubit if $1$ is measured results in $\ket{b} \ket{b}$. The last measurement and conditional gate gives either $\ket{b}$ or the initial $\ket{\psi}$ state. The three steps in Fig. \ref{fig:Teleportation} can be represented as 
\begin{equation*}
\ket{b} \ket{a} \ket{a} \rightarrow \ket{b} \ket{a \oplus b} \ket{a} \rightarrow \ket{b} \ket{b} \rightarrow \ket{b}.
\end{equation*}

While this both inputs and outputs $\ket{\psi}$, the advantage is that the output can be a post-selected logical qubit, which can be used to correct errors in the physical qubits that comprise an arbitrary logical data qubit.
\figTeleportation

Now, if we have the $3$ qubit cat state
\begin{equation*}
\frac{1}{\sqrt{2}}(\ket{000} + \ket{111}) = \frac{1}{\sqrt{2}} \sum_{a=0}^1 \ket{aaa},
\end{equation*}
and apply a similar process shown in Fig. \ref{fig:TeleportationSplit}, then after each of the $3$ steps, we have 
\begin{equation*}
\ket{b} \ket{a} \ket{a} \ket{a} \rightarrow \ket{b} \ket{a \oplus b} \ket{a} \ket{a} \rightarrow \ket{b} \ket{b} \ket{b} \rightarrow \ket{b} \ket{b}.
\end{equation*}
The effect of this is to send $\ket{0}$ to $\ket{00}$ and $\ket{1}$ to $\ket{11}$ with freshly post-selected qubits.
\figTeleportationSplit

In the Hadamard basis, the $3$ qubit cat state is transformed to
\begin{equation*}
\frac{1}{2} (\ket{000} + \ket{011} + \ket{101} + \ket{110}) = \frac{1}{2} \sum_{a=0}^1 \sum_{b=0}^1 \ket{a} \ket{b} \ket{a \oplus b}.
\end{equation*}
In Fig. \ref{fig:TeleportationMerge}, a circuit is shown that uses $2$ teleportations to bring in two additional states with the basis elements $\ket{c}$ and $\ket{d}$. After each step in this circuit, this results in
\begin{align*}
\ket{c} \ket{d} \ket{a} \ket{b} \ket{a \oplus b} \rightarrow \ket{c} \ket{d} \ket{a \oplus c} \ket{b} \ket{a \oplus b} \\
\rightarrow \ket{c} \ket{d} \ket{a \oplus c} \ket{b \oplus d} \ket{a \oplus b} \rightarrow \ket{c} \ket{d} \ket{b \oplus d} \ket{c \oplus b} \\
\rightarrow \ket{c} \ket{d} \ket{c \oplus d} \rightarrow \ket{d} \ket{c \oplus d} \rightarrow \ket{c \oplus d}.
\end{align*}
\figTeleportationMerge

If these two previous ancillas are used together, they give 
\begin{equation*}
\ket{a} \ket{b} \rightarrow \ket{a} \ket{a} \ket{b} \rightarrow \ket{a} \ket{a \oplus b},
\end{equation*}
which is {\tt CNOT}. Alternately, we can combine the two ancillas via teleportation to get the $4$ qubit state 
\begin{equation*}
\ket{a} \ket{b} \ket{a} \ket{a \oplus b} = \frac{1}{2}(\ket{0000} + \ket{1011} + \ket{0101} + \ket{1110}),
\end{equation*}
 and teleport $\ket{a}$ and $\ket{b}$ into the first and second qubits respectively as illustrated  in Fig. \ref{fig:CNOTTeleportation}.

\figCNOTTeleportation

For CSS codes, we can assume that noise from a logical {\tt CNOT} acting on two logical data qubits determines the fault tolerant thresholds. See section \ref{othergates} for a discussion of the other gates necessary for universal quantum computation.

\subsection{Repeated post-selection}
This section discusses the process of repeated post-selection. By performing a large number of post-selections, the resulting noise decreases in magnitude, converging to some value. 

The process which uses ancilla qubits can be broken down into two parts. First, we run post-selection repeatedly to refine ancilla qubits. This process works by detecting bit flip errors as in Fig. \ref{fig:ErrDetect}, then switching to the Hadamard basis and again applying the circuit of Fig. \ref{fig:ErrDetect} to detect what were originally phase flip errors.  As post-selection is repeatedly applied, the noise will converge to some fixed channel that can be estimated by iterative calculation. Second, we use teleportation to bring a logical data qubit into the ancilla qubits; simultaneously this applies a logical {\tt CNOT}. 

\paragraph{Noise under post-selection}
When we are running error detection with bit flip errors, we apply {\tt CNOT} transversally from the source qubits to the destination qubits. Bit flip errors will propagate from the first to the second qubit, just like we want them to. However, phase flip errors will propagate backwards. To deal with this, we minimize the phase flip errors in the destination qubits. The source qubits will also have been post-selected for phase flip errors in the previous step. Therefore, before we apply the {\tt CNOT}, we can assume that there is the same diagonal noise on the source and destination qubits. We let $\sigma = S_\sigma = D_\sigma$.  Then using Eq. \ref{eq:measure}, the noise after {\tt CNOT} followed by a post-selected measurement (measurement with no error) becomes
\begin{align*}
x_{\text{out}} = \frac{x^2 Q_{XI} + y^2 Q_{XZ}}{1 + z^2 Q_{IZ}}\\
y_{\text{out}} = xy \frac{Q_{YI} + Q_{YZ}}{1 + z^2 Q_{IZ}}\\
z_{\text{out}} = z \frac{Q_{ZI} + Q_{ZZ}}{1+z^2 Q_{IZ}}.
\end{align*}

In the next step, we run error detection on the other type of noise (bit flip or phase flip). This can be represented by applying a Hadamard gate to change bit flip errors into phase flip errors and vice versa, and then applying the same error detection process. This gives $x_{out} = z$, $z_{out}=x$, $y_{out} = y$, yielding the map
\begin{equation}
\label{eq:postselectmap}
[x,y,z] \rightarrow [z_{\text{out}}, y_{\text{out}}, x_{\text{out}}].
\end{equation}

Since we just post-selected for bit flip errors, after the Hadamard there is a low rate of phase flip errors $p_f = p_X + p_Y$. From Eq. \ref{eq:chanpauli}, we see that $x = 1 - 2p_f$ is close to $1$.

Under repeated post-selection, the noise converges to a solution $[1,x,y,z]$ of 
\begin{align*}
(1 + z^2 Q_{IZ})^2 = z (Q_{YI}+Q_{YZ}) (Q_{ZI}+Q_{ZZ})\\
x = \frac{1 + z^2 Q_{IZ}}{Q_{YI}+ Q_{YZ}}\\
y^2 Q_{XZ} = z (1+z^2 Q_{IZ}) - x^2 Q_{XI}. 
\end{align*}
This is the effective noise after infinity iterations of post-selection.

\paragraph{Teleportation}
Once there has been sufficient post-selection on the logical source qubit, we apply {\tt CNOT} from it to a logical destination qubit, and then measure the physical destination qubits in the normal $Z$ basis (which detects bit flip errors) and the source qubits in the $X$ basis (which detects phase flip errors) as shown in Fig. \ref{fig:Teleportation}. It is important that the two logical qubits which will have a {\tt CNOT} applied between them are post-selected last for opposite types of errors, one phase flip, one bit flip. This  results in bit flip errors and phase flip errors being equally probable --- useful for CSS codes which are very inefficient at correcting noise of only one of those types at a time. 

Therefore, $S_Y = D_Y$, $S_X = D_Z = x$ (or $z$) $S_Z = D_X = z$ (or $x$). Consider a pair of physical qubits that are about to undergo an $X$ measurement and a $Z$ measurement as part of the teleportation circuit shown in Fig. \ref{fig:Teleportation}. A bit flip error ($X$ or $Y$ errors) in the destination qubit results in an $X$ error. A phase flip error ($Z$ or $Y$ errors) in the source qubit results in a $Z$ error. If both occur, we have a $Y$ error; and if neither occur, then no error. Thus, the probabilities of Eq. \ref{eq:Teleporterror1} are
\begin{equation}
\label{eq:Teleporterror3}
\begin{bmatrix}
p_I \cr
p_X \cr
p_Y \cr
p_Z \cr
\end{bmatrix} = \frac{1}{4} \begin{bmatrix}
1 & 1  &  1 & 1 \cr
1 & 1  & -1 & -1 \cr
1 & -1 &  1 & -1 \cr
1 & -1 & -1 & 1 \cr
\end{bmatrix} \begin{bmatrix}
1 \cr
m xz Q_{XI} \cr
m^2 y^2 Q_{XZ} \cr
m xz Q_{IZ} \cr
\end{bmatrix},
\end{equation}
where $m$ is defined in terms of the probability of a measurement error in Eq. \ref{eq:m}. 

Since we assume independent noise on each qubit, we can judge how well the quantum code corrects this noise by using the optimal recovery adaptive concatenation technique of \cite{MoreChans}. These are calculated in the next section.

\begin{table}
\caption{\label{tab:errortypes} The probability of each of the $16$ two qubit Pauli errors resulting from a {\tt CNOT} gate for depolarizing noise and independent noise. Note that the probability of a measurement error $p_m$ is a separate parameter. The depolarizing noise used in this paper has no measurement errors; Knill noise is depolarizing noise with a  probability $p_m=\frac{4}{15}p$ of a measurement error.  Forward type noise has no measurement errors.}
\begin{ruledtabular}
\begin{tabular}{cccc}
Prob.       & Depolarizing   & Forward    & Independent noise \\
\hline
$p_{II}$    & $1-p$          & $(1-p_f)^2$  & $(1-p_f)^2(1-p_b)^2$\\
$p_{IX}$    & $\frac{p}{15}$ & $(1-p_f)p_f$ & $(1-p_f)p_f(1-p_b)^2$\\
$p_{IY}$    & $\frac{p}{15}$ & $(1-p_f)p_f$ & $(1-p_f)p_f(1-p_b)p_b$\\
$p_{IZ}$    & $\frac{p}{15}$ & $(1-p_f)^2$  & $(1-p_f)^2(1-p_b)p_b$\\
$p_{XI}$    & $\frac{p}{15}$ & $(1-p_f)^2$  & $(1-p_f)^2(1-p_b)p_b$\\
$p_{XX}$    & $\frac{p}{15}$ & $(1-p_f)p_f$ & $(1-p_f)p_f(1-p_b)p_b$\\
$p_{XY}$    & $\frac{p}{15}$ & $(1-p_f)p_f$ & $(1-p_f)p_fp_b^2$\\
$p_{XZ}$    & $\frac{p}{15}$ & $(1-p_f)^2$  & $(1-p_f)^2p_b^2$\\
$p_{YI}$    & $\frac{p}{15}$ & $(1-p_f)p_f$ & $(1-p_f)p_f(1-p_b)p_b$\\
$p_{YX}$    & $\frac{p}{15}$ & $p_f^2$      & $p_f^2(1-p_b)p_b$\\
$p_{YY}$    & $\frac{p}{15}$ & $p_f^2$      & $p_f^2 p_b^2$\\
$p_{YZ}$    & $\frac{p}{15}$ & $(1-p_f)p_f$ & $(1-p_f) p_f p_b^2$\\
$p_{ZI}$    & $\frac{p}{15}$ & $(1-p_f)p_f$ & $(1-p_f)p_f(1-p_b)^2$\\
$p_{ZX}$    & $\frac{p}{15}$ & $p_f^2$      & $p_f^2(1-p_b)^2$\\
$p_{ZY}$    & $\frac{p}{15}$ & $p_f^2$      & $p_f^2(1-p_b)p_b$\\
$p_{ZZ}$    & $\frac{p}{15}$ & $(1-p_f)p_f$ & $(1-p_f)p_f(1-p_b)p_b$\\
\end{tabular}
\end{ruledtabular}
\end{table}

\begin{table}
\caption{\label{tab:hashingfault} Noise that gives the hashing bound (Shannon entropy of $1$) under repeated post-selection. The $p_{\text{in}}$ is the input rate of errors for the types of noise as defined in this paper. The $p_{\sigma}$ are the probabilities of a Pauli error $\sigma$ on each qubit as given in Eq. \ref{eq:Teleporterror3}. \noises These values were determined by numerical calculations.}
\begin{ruledtabular}
\begin{tabular}{cccc}
            & Depolarizing & Knill & Forward \\
\hline
$p_{\text{in}}$ & 8.27515\% & 6.90240\% & 4.81816\% \\
$p_X=p_Z$       & 7.13361\% & 7.52699\% & 9.79217\% \\
$p_Y$           & 4.78136\% & 4.12990\% & 1.21061\%
\end{tabular}
\end{ruledtabular}
\end{table}

\begin{table}
\caption{\label{tab:thresholds} Fault tolerant thresholds in terms of $p_{\text{in}}$  assuming independent noise on each qubit. \noises}
\begin{ruledtabular}
\begin{tabular}{cccc}
Code         & Depolarizing & Knill & Forward \\
\hline
Hashing      & 8.2751\%   & 6.9024\%   & 4.8182\% \\
$[[7,1,3]]$  & 8.229(7)\% & 6.864(5)\% & 4.8036\%\\
$[[17,1,5]]$ & 8.2(0)\%   & 6.8\%      & 4.790\% \\
$[[23,1,7]]$ & 8.25\%     & 6.88\%     & 4.805\% \\
\end{tabular}
\end{ruledtabular}
\end{table}

\begin{table}
\caption{\label{tab:capacity} Capacity of codes for $(p_X,p_Y,p_Z)$ type noise}
\begin{ruledtabular}
\begin{tabular}{ccc}
Code & $(0,0,p)$ & $(p,p,p)$ \cr
\hline
Hashing      & 11.0028\%    & 6.3097\% \\
$[[7,1,3]]$  & 10.963(2)\%  & 6.270(4)\% \\
$[[17,1,5]]$ & 10.927(0)\%  & 6.251 \% \\
$[[23,1,7]]$ & 10.968\%     & 6.29\%  \\
$[[4,2,2]]$ and $[[6,2,2]]$ & 10.9466\% & 6.271(9)\%
\end{tabular}
\end{ruledtabular}
\end{table}

\section{Threshold calculations for uncorrelated noise on each qubit} 
\label{sec:noisemodels}
The probabilities $p_{\sigma \sigma'}$ of a Pauli error $\sigma \sigma'$ for the types of noise resulting from {\tt CNOT} that are discussed in this paper are given in Tab. \ref{tab:errortypes}. 

\subsection{Depolarizing noise}
Suppose that the {\tt CNOT} gate has probability $\frac{p}{15}$ of each of the non-identity two qubit Pauli errors. The total probability of no error is therefore $1-p$. This results in $Q_{\sigma \sigma'} = 1 - \frac{16}{15} p$, except for $Q_{II} = 1$.

We look at two possibilities for measurement errors. We refer to the case where there are no measurement errors as depolarizing noise. In this case, $p_m = 0$ and $m =1$. Knill noise \cite{Knill,Knill2} refers to the case where there is probability $p_m = \frac{4}{15} p$ of a measurement error, which results in $m = 1-\frac{8}{15} p$. 

We first assume the same independent noise on each qubit, and calculate the resulting such noise from Eq.  \ref{eq:Teleporterror3}. We find the noise corresponding to the Hashing bound for these noises in Tab. \ref{tab:hashingfault}, that is, when the Shannon entropy of the noise is $1$. We then determine whether a given code can correct this noise under repeated concatenation with itself by using the optimal recovery operator method of \cite{MoreChans}, which considers syndrome information from previous levels to determine the optimal recovery operator at each level, and uses Monte Carlo simulation to estimate the threshold. The stabilizers for the codes studied in Tab. \ref{tab:hashingfault} are given in \cite{GrasslCSS}. 

In Fig. \ref{fig:depol}, we plot the effect of the measurement errors $p_m = \frac{4}{15} r p$ (that is, the fraction $r$ of the Knill measurement error) on the fault tolerant threshold (assuming the hashing bound). 

\begin{figure*}
\includegraphics{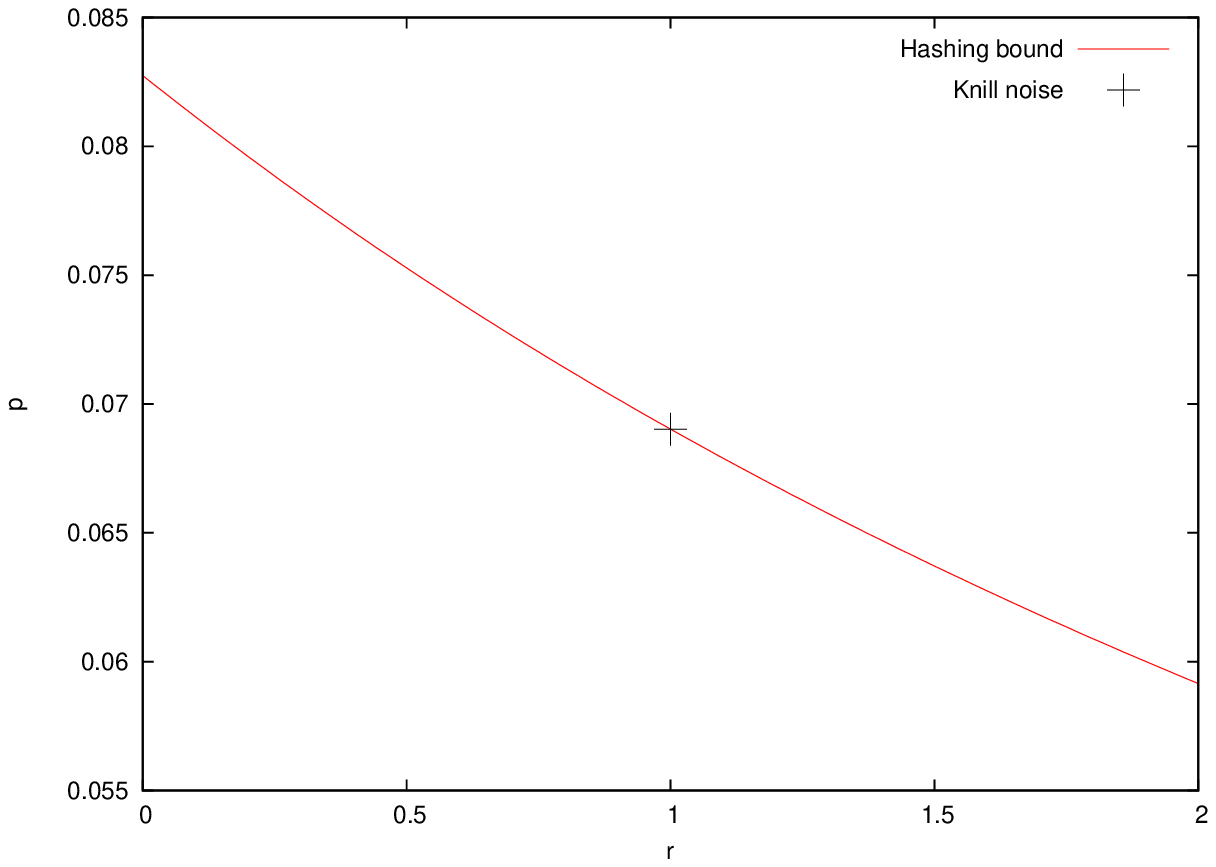}
\caption{\label{fig:depol}
The hashing bound for fault tolerant depolarizing noise. The {\tt CNOT} gate has probability $\frac{p}{15}$ of each non-identity two qubit Pauli error occurring. The probability of a measurement error is $p_m = \frac{4}{15} r p$, where the $r=1$ gives the measurement error used by Knill \cite{Knill}. }
\end{figure*}

\subsection{Independent bit flip and phase flip noise}
\label{sec:uncorrelated}

In this section, we consider noise from the {\tt CNOT} gate, where bit flips and phase flips act independently of each other, and the resulting noise on the source and destination qubits are also independent of each other. Note that since bit flip and phase flip errors are independent, their error detection and correction are done independently of each other. The calculations in this section can be performed directly from Eqs. \ref{eq:postselectmap} and \ref{eq:Teleporterror3}; however, because they are simpler in the case discussed here, we will go into more detail. 

The probability of a phase error on the source qubit and the probability of a bit flip error on the destination qubit are both designated $p_f$. The probability of a bit flip error on the source qubit and the probability of a phase flip error on the destination qubit are both designated $p_b$. The probability of a measurement error is $p_m$. The probabilities of each of the Pauli errors for this type of noise are calculated in Table \ref{tab:errortypes}. For example, since there is probability $1-p_f$ of no bit flip error on the destination qubit, probability $1-p_f$ of no phase flip error on the source qubit, probability $1-p_b$ of no phase error on the destination qubit, and probability $p_b$ of a bit flip error on the source qubit, it follows that $p_{XI}=(1-p_f)^2 (1-p_b) p_b$.

Let $p_g$ be the probability of a bit flip or phase flip error immediately after post-selection is performed for that type of error. Let $p_b$ be the probability of an error before that post-selection is run, i.e. after the other type of error detection is run.

We use the noise notation from Appendix \ref{sec:noise}. Before we run post-selection, we have the noise $x_b = 1 - 2p_b$ on both the source and destination qubits. After we apply ${\tt CNOT}$, the destination qubit gains the noise $f = 1 - 2 p_f$ representing the probability of forward type error from {\tt CNOT}. Measuring it results in the noise $m = 1 - 2p_m$, representing the probability of a measurement error. Therefore the total destination qubit noise is $x_g f m$. After post-selection, we must add the backwards type noise from {\tt CNOT} to the source qubit. From Sec. \ref{sec:noise}, we have
\begin{equation*}
x_g = b \text{ post} (x_b, x_b f m) = b \frac{x_b + x_b f m}{1 + x_b^2 f m}.
\end{equation*}

When we run post-selection on the other type of noise (bit flip or phase flip), the $x_g$ noise of this type back-actions from the ancilla qubits. In addition, the forward type noise contributes $f$ noise, and so the total noise after the other type of error detection becomes
\begin{equation*}
x_b = x_g^2 f = b^2 f x_b^2 (\frac{1 + fm}{1+x_b^2 f m})^2
\end{equation*}
before we run error detection on this type of noise (either bit flip or phase flip) again.

After repeated error detection and post-selection, the noise will tend towards an equilibrium value, which we calculate iteratively.

After the logical qubits have been sufficiently refined, we start the teleportation process given in Fig. \ref{fig:CNOTTeleportation}. Applying {\tt CNOT} results in a back-action. To improve the threshold, we have the source and data qubits optimized for opposite types of errors, so the back-action together with the regular noise results in $x_g x_b$ noise. In addition, the {\tt CNOT} gate gives $f$ noise, and the actual measurement gives $m$ noise, resulting in a total combined noise from teleportation of
\begin{equation}
\label{eq:combinednoise}
x=x_b x_g f m = x_g^3 f^2 m.
\end{equation}

For CSS codes, the measurement process detects both $X$ and $Z$ errors, and is equally good at correcting both. If there are independent probabilities $p$ of a bit flip and a phase flip error, then $(p_X, p_Y, p_Z) = (p-p^2,p^2,p-p^2)$. We consider for which types of noise the code can correct, assuming no additional errors during error detection and correction, as  in \cite{MoreChans}. For CSS codes, this has the same threshold as noise in the form $(p_X,p_Y,p_Z)=(0,0,p)$. In Tab. \ref{tab:capacity}, we calculate the thresholds for various CSS codes; these are found to be slightly below the hashing bound of $p=11.0028\%$ found in \cite{MoreChans}.

We define forward type noise to consist of probability $p_f$ of forward noise,with back noise probability $p_b$ and measurement noise probability $p_m$ both $0$. Therefore, $n=m=1$ and $f=1-2p_f$. At the hashing bound, the noises are $f=0.90363682$ (probability $p_f=4.818159\%$ of a forward type noise), $x_g = 0.98482389$ and $x_b = 0.87641757$. The combined hashing bound noise is $c=x_g x_b f = x_g^3 f^2 = 0.77994427$, that is $p_c = \frac{1-c}{2}$ is the hashing bound. 

\ignore{
\begin{table}
\caption{\label{tab:degeneracy} Calculation of how much degeneracy decreases the threshold for two CSS codes}
\begin{ruledtabular}
\begin{tabular}{ccc}
Variable       & $[[7,1,3]]$      & $[[23,1,7]]$ \\
\hline
minimum weight         &    4             &  8 \\
number                 &    7             & 506 \\
undetected crash ways  &  7 \binom{4}{2}  & 506 \binom{8}{4} \\
total undetected prob. & 3(7 \binom{4}{2}
prob crash             & \binom{7}{2} p^2 & \binom{23}{4} p^4 \\

Hashing      & 11.0028\%   & 6.3097\% \\
$[[7,1,3]]$  & 10.963(0)\% & 6.270(3)\% \\
$[[17,1,5]]$ & 10.927\%    & 6.2(5) \% \\
$[[23,1,7]]$ & 10.968\%    & - \\
\end{tabular}
\end{ruledtabular}
\end{table}
}

\section{Entropy-based correlated noise calculations}
\label{sec:entropy}
In the previous parts of this paper, we assumed that there is the same local uncorrelated noise on each qubit. However, degeneracies in the code caused by low weight stabilizer elements will introduce correlated errors that are noisier than the noise predicted by the uncorrelated noise model. The most significant degeneracies occur at the first level of the code. In this section, we calculate the actual noise after post-selection at the first level (and sometimes look at the weaker effects at the second level) of the code. We use the entropy of this resulting correlated noise to evaluate the fault tolerant threshold. We then compare these exact results to those of the uncorrelated case.  We see how the differences result from the degeneracies that arise from the low weight stabilizer elements. 

\subsection{Entropy of the $[[7,1,3]]$ code with forward type noise}
The $[[7,1,3]]$ code is discussed in Appendix \ref{sec:713}. 

Since the bit flip and phase flip errors are corrected independently of each other for this type of noise, we only have to consider correction of the one type of error. There can either be no error detected, in which case there was either a distance $0$ or a distance $3$ error, or there are $7$ equivalent syndromes corresponding to the detection of an error, in which case a distance $1$ or a distance $2$ error occurred. If $p_n$ are the probabilities of distance $n$ errors, then the Shannon entropy is
\begin{align*}
(p_0+p_3)(f(\frac{p_0}{p_0+p_3}) + f(\frac{p_3}{p_0+p_3}))\\
+(p_1+p_2)(f(\frac{p_1}{p_1+p_2}) + f(\frac{p_2}{p_1+p_2}))\\
= \sum_n f(p_n) - f(p_0 + p_3) - f(p_1 + p_2),
\end{align*}
where
\begin{equation*}
f(x) = -x \log_2 x.
\end{equation*}

We use the results of Sec. \ref{sec:713} to perform calculations in terms of the syndromes of the $[[7,1,3]]$ code that are similar to that of Sec. \ref{sec:uncorrelated}.

Tab. \ref{tab:capacity} shows that the $[[7,1,3]]$ CSS code can correct $(0,0,p)$ type noise for $p < 10.963(2)\%$. This noise has an entropy of $e_1=\frac{1}{2} 1.0020(5)$ after the first level, and $e_2=\frac{1}{2} 1.0011(8)$ after the 2nd level.

We use the processes described in this section to calculate what happens to one level of the $[[7,1,3]]$ code in terms of the probabilities of a distance $d$ error, instead of assuming the same noise on each qubit. We calculate for which value of $p_f$ the entropy after the second level of the code is the same $e_2$ as before, in the uncorrelated case. This gives a threshold for forward type noise of $p_f= 4.6700\%$. As a check, we note that the entropy after the 1st level was $\frac{1}{2} 1.0020(8)$, almost the same as before. By repeatedly concatenating the resulting noise for several levels, calculations have confirmed that this is indeed the new threshold.

\subsection{Analysis of the difference}
\label{sec:diff}
In the previous section, we calculate that the actual threshold for forward type noise for the $[[7,1,3]]$ code is $p_f=4.6700\%$. This compares to the original approximation of $p_f = 4.8036\%$ found in Tab. \ref{tab:thresholds}. 

The difference occurs because the code is slightly degenerate; there exist cases where the error that would be measured in the destination ancilla qubit corresponds to a stabilizer element, and therefore is not detected as an error. From Tab. \ref{tab:713dist}, it can be seen that the $[[7,1,3]]$ code has $7$ stabilizer elements of weight $4$. Thus, if there are $2$ errors in the source qubits and $2$ errors in the destination qubits, there can be an undetected weight $2$ error in the source qubits. For each stabilizer element, this has an approximate probability  $\binom{4}{2} p_g^2$ of occurring, where $p_g$ is the probability of an undetected error from Sec. \ref{sec:uncorrelated}. Together these give a factor of $42p_g^2$ difference.

In Eq. \ref{eq:combinednoise}, we see that errors left undetected (and uncorrected) after post-selection contribute $3$ times to the combined rate of noise analyzed for correctability by the code. Therefore, the total estimated difference in the probability of an encoded error (crash probability)  is $c_e = 126p_g^2$.

If $p_f = 4.6700\%$, then $p_g = 0.70\%$, and so the estimated difference is $c_e = 0.62\%$. To calculate the actual difference, we understand that if there is the same noise $x$ on each qubit, the resulting noise map for an encoded crash is $f_7(x)=\frac{7}{4}x^3 - \frac{3}{4}x^7$ \cite{FKSS}. For $p_f=4.6700\%$, $x=0.78794(5)$, and $f_7(x)=0.7147$. For $p_f=4.8036\%$, $x=0.78073(6)$, and $f_7(x)=0.7002$. The difference in crash probabilities is half the difference of $f_7(x)$, and is $c_a = 0.72\%$. 

\subsection{Further calculations}

\begin{table}
\caption{\label{tab:analysis} Analysis of difference for $[[7,1,3]]$ code. $p_a$ are the actual thresholds found from analyzing the correlated noise resulting from post-selection at the first level of the code. (However, later we find that the threshold for the forward noise is $p_a=4.6699\%$.) \noises}
\begin{ruledtabular}
\begin{tabular}{cccc}
Value          & Depolarizing  & Knill      & Forward  \\
\hline
$p_e$          & 8.229(7)\%    & 6.864(5)\% & 4.8036\%\\
$p_a$          & 8.1096\%      & 6.7785\%   & 4.6700\%\\
$e_2$          & 1.0012(1)     & 1.0011(9)  & 1.0011(7) \\
$p_s$          & 5.58\%        & 4.69\%     & 5.93\%  \\
$p_d$          & 9.42\%        & 9.48\%     & 10.04\%  \\
$p_g$          & 0.61\%        & 0.51\%     & 0.70\%  \\
$c_e$          & 0.47\%        & 0.33\%     & 0.62\%  \\
$c_a$          & 0.39\%        & 0.32\%     & 0.72\%\\
\end{tabular}
\end{ruledtabular}
\end{table}

In Tab. \ref{tab:analysis}, we perform similar calculations for depolarizing and Knill noise. $p_e$ is the threshold estimated by assuming the same noise on each qubit. $e_2$ is the logical entropy at the 2nd level of the code using this noise. $p_a$ is the actual threshold found by fully post-selecting at the first level of the $[[7,1,3]]$ CSS code. Using the same noise on each qubit with the actual probability of an error $p_a$: $p_s$ and $p_d$ are the probabilities of source and destination qubit errors immediately before post-selection. $p_g$ is the probability of an undetected error after post-selection. $c_e = 126 p_g^2$ is the estimated difference in crash probabilities for a bit flip error ($X$ or $Y$). Phase flip errors ($Y$ or $Z$) have the same probabilities.  $c_a$ is the actual (not the estimated) difference in the crash probabilities due to the degeneracies from the stabilizer elements. 

In each of the cases, the noise at threshold for the actual non-local noise $p_a$ is found to have the same $e_2$ as the $p_e$ noise, with the same noise on each qubit.  It is interesting to note that for all of these cases, the entropy at the 2nd level of the code, $e_2$, is the same; therefore to determine the threshold for a particular type of noise for the $[[7,1,3]]$, one just has to find when the entropy of the logical noise at the second level of the code is $e_2 = 1.0011(8)$. 

\begin{table}
\caption{\label{tab:thresholdvalues} Values at the hashing bound.  \noises}
\begin{ruledtabular}
\begin{tabular}{cccc}
Value          & Depolarizing  & Knill      & Forward  \\
\hline
$p$            & 8.27515\%     & 6.90240\%  & 4.81816\%\\
$p_s$          & 5.65\%        & 4.80\%     & 6.18\%  \\
$p_d$          & 9.56\%        & 9.67\%     & 10.40\%  \\
$p_g$          & 0.63\%        & 0.54\%     & 0.76\%  \\
\end{tabular}
\end{ruledtabular}
\end{table}

So far we have only considered the non-local effects of post-selection from the first level. What about the second level? In addition to the stabilizers in each of the sub-blocks from the previous level of the code, there are new stabilizers that consist of an encoded $\overline{X}$ in $4$ of the $7$ sub-blocks. The lowest weight for these is $12$, which consist of distance and weight $3$ errors in each of $4$ sub-blocks. There are $7$ of these stabilizers at the top level of the code, so there are $7$ ways to choose the sub-blocks. Each of the sub-blocks has $7$ elements of distance and weight $3$. Therefore there are $7^5$ stabilizers of distance $12$. The estimated difference in the crash probability is $c_e = 3 \times 7^5 \binom {12}{6} p_g^6$. 

For the forward type noise $c_e = 6.5 \times 10^{-6}$, which results in a difference in $p_f$ of less than $10^{-6}$, so the final actual threshold for the forward type noise for the  $[[7,1,3]]$ code is $p_f = 4.6699\%$. The effect on the other types of noise is insignificant.

\subsection{The $[[23,1,7]]$ CSS code}
The $[[23,1,7]]$ CSS code has $506$ weight $8$ elements corresponding to no error \cite{BWR}, i.e., $506$ weight $8$ stabilizer elements. Using the method of the previous section, the estimated difference in crash probability is 
\begin{equation*}
c_e = 3 \binom{8}{4} 506 p_g^4.
\end{equation*}
For the forward type noise near the hashing bound, this results in $c_e = 0.00035$. 

Again, using the noise notation of Appendix \ref{sec:noise}, the undetected distance $7$ elements have weights $7$, $11$, $15$, and $23$ \cite{BWR}, so the noise map for a logical crash is of the form  $f_{23}(x) = c_7 x^7 + c_{11} x^{11} + c_{15} x^{15} + c_{23} x^{23}$. Since no error input results in no error output, $f(1)=1$. Since this code can correct any distance $1$, $2$, or $3$ error, it follows that $f'(1)=f''(1)=f'''(1)=0$. Solving these gives the channel map of 
\begin{equation*}
f_{23}(x) = \frac{3795}{512}x^7 -\frac{805}{64}x^{11} + \frac{1771}{256}x^{15} - \frac{385}{512}x^{23}.
\end{equation*}
The crash probability is $\frac{1-f_{23}(x)}{2}$. For the same noise on each qubit, we found the threshold to be $p_e =4.805(4)\%$, and $x=0.78064$. Solving for when the difference in crash probability is $\Delta = 0.00035$, we get $p_r= 4.801(4)\%$, only lowering the threshold by $4 \times 10^{-5}$ in $p_f$. The actual threshold was calculated to be $p_a=4.800(3)\%$. 

\begin{table}
\caption{\label{tab:2317thresholdvalues} $[[23,1,7]]$ thresholds. \noises}
\begin{ruledtabular}
\begin{tabular}{cccc}
Value          & Depolarizing  & Knill      & Forward  \\
\hline
$p_e$          & 8.25(3)\%     & 6.88(2)\%  & 4.805(4)\%  \\
$c_e$          & 0.00017       & 0.00009    & 0.00035   \\
$\Delta_p$     & 0.003\%       & 0.002\%    & 0.0040\%  \\
$p_a$          & -             & -          & 4.800(3)\%\\
$p_r$          & 8.25(0)\%     & 6.88(0)\%  & 4.801(4)\%  \\
\end{tabular}
\end{ruledtabular}
\end{table}

These calculations were also performed for Knill and depolarizing noise in Tab. \ref{tab:2317thresholdvalues}. For the Forward type noise, $p_e=4.805(4)\%$, which results in a total noise of $p=10.968\%$, the threshold for one type of noise for the $[[23,1,7]]$ code. The logical entropy after the first level is $e_1=1.0016\%$. For the other types of noise, the values of $p_e$ were calculated by determining which noise gave this same logical entropy of $e_1$ after the first level. Based upon calculations with change in entropy near the threshold under different levels of concatenation for the $[[7,1,3]]$ and $[[17,1,5]]$ CSS codes, this method should be accurate to at most a couple times $10^{-5}$ for the entropy and a small fraction of $10^{-5}$ for the thresholds.

\subsection{Encoding and post-selection process}
\label{encoding}p
\figNNCNOT
The methods described earlier in this paper rely on the code being concatenated many times with itself. A possible problem occurs if we create encoded $\overline{\ket{0}}$ for the $[[7,1,3]]$ code like in Fig. \ref{fig:encode}, and post-select as described in this paper, then the $Z$ noise will converge to the completely depolarizing case, where each of the $8$ error syndromes are equally likely. The reason for this problem is that $\overline{Z}$ is now a stabilizer, and there are now $7$ stabilizer elements of weight $3$.

To correct this problem, instead of post-selecting just one logical qubit after the first level of the $[[7,1,3]]$ code, we first create a logical $\overline{\ket{0}}$ state for the $[[7,1,3]]$ code as in Fig. \ref{fig:encode} and again in the Hadamard basis, so that we have the logical states $\overline{\ket{0}}$ and $\overline{\ket{+}}$. We apply a logical {\tt CNOT} transversally from each of the physical qubits of the $\overline{\ket{+}}$ state to the $\overline{\ket{0}}$ state to create an encoded Bell pair $\frac{1}{\sqrt{2}} \overline{\ket{00} + \ket{11}}$ that uses $14$ physical qubits. Now there are no longer any stabilizer elements of weight $3$. The $7$ weight $4$ ones were discussed in Sec. \ref{sec:diff}. The next lowest weight stabilizers are the $49$ weight $6$ stabilizers that correspond to $\overline{Z} \otimes \overline{Z}$.

Post-selection can now be used to refine this state. To post-select, we start with two Bell pair states. We apply logical {\tt CNOT} from the first logical qubit of the first Bell pair to the second logical qubit of the second Bell pair. Similarly, we apply logical {\tt CNOT} from the second logical qubit of the first Bell pair to the first logical qubit of the second Bell pair. When applying the logical {\tt CNOT} gates, we mix up the orders of the qubits, to even out the errors, which is better for post-selection. There are 168 ways to permute the $7$ qubits of the $[[7,1,3]]$ CSS code and preserve the code. let $p(x)$ represent such a permutation. The distance $3$ codewords are those qubits $q_1$, $q_2$ and $q_3$ such that their parity is even, that is, $q_1 \oplus q_2 \oplus q_3 = 0$. The parity of $p(q_1)$, $p(q_2)$ and $p(q_3)$ is also even, and $p(1)$, $p(2)$ and $p(4)$ determine the permutation. 

We then measure the second Bell pair in the standard basis, rejecting the first Bell pair if an error is detected. We apply the Hadamard gate transversally to the remaining (first) Bell pair, to exchange bit flip errors and phase flip errors for the next round of post-selection. Notationally, we consider the initial encoding of the Bell pair to be the $0$ times post-selected states. To create the $n$ times post-selected state, we use an $n-1$ times post-selected state as the source (1st) Bell pair, and a $d_n$ times post-selected state for the destination (2nd) Bell pair. $d_n = 0, 0, 1, 2, 3, \ldots$ is a well-performing sequence for post-selection. 

After a sufficient amount of post-selection is performed, we take $7$ Bell pairs, and measure one of the logical qubits in the Hadamard basis on the 2nd qubit, rejecting if there is an error. We now have $7$ post-selected encoded qubits. We use these qubits to create the $\overline{\ket{0}}$ state for the  $[[7,1,3]]$ code encoded with itself (using a total of $49$ qubits), and use more Bell pairs to detect if there are any errors after each step, rejecting everything if an error is detected. We then construct an encoded Bell pair, using $98$ physical qubits, and repeat the process described above for post-selection, rejecting everything if there is an error detected in a single qubit. 

Once we have done this for a number of levels and the rate of errors is sufficiently low, we have two options. One is that we continue to use the $[[7,1,3]]$ code, and we construct the more complicated encoded ancilla shown in Fig. \ref{fig:CNOTTeleportation}. The other possibility is that we prefer to use some other code, such as the $[[23,1,7]]$ CSS code. In this case, we now encode into that code. However, we wish to remove the levels of the $[[7,1,3]]$ CSS code from the encoding, since the $[[7,1,3]]$ code might not have as good a threshold (when we actually use error correction, not this post-selection process). We can decode one level at a time, using the following process for decoding the $[[7,1,3]]$ CSS code. Suppose we have the encoded state $c_0 \overline{\ket{0}} + c_1 \overline{\ket{1}}$. We can represent this as
\begin{align*}
m_1= x_1 \oplus e &&  m_2 = x_2 \oplus e \\
m_3= x_2 \oplus x_1 \oplus e && m_4 =  x_4 \oplus e\\
m_5 = x_4 \oplus x_1 \oplus e && m_6= x_4 \oplus x_2 \oplus e\\
m_7= x_4 \oplus x_2 \oplus x_1 \oplus e,
\end{align*}
where $x_1,x_2,x_4$ are summed over $0$ and $1$, and $e$ represents the encoded basis states. We post-select it for $Z$ errors. Now we measure the last $4$ qubits, rejecting everything and starting over if the parity isn't even. $m_4 \oplus m_5 = x_1$ and $m_4 \oplus m_6 = x_2$, so by applying $X$ gates to the first $3$ qubits conditional on these measurements, we have $m_1=m_2=m_3=e$, which is $c_0 \ket{000} + c_1 \ket{111}$. We then measure two of the qubits in the Hadamard basis. To avoid introducing more error, we can reject if either one is in the $\ket{-}$ state. For the forward type noise near the hashing bound, the total bit flip error is about $3p_g \approx 2\%$, and the total phase flip error is about $p_b \approx 6\%$. Once we have removed the lowest level of the $[[7,1,3]]$ code, we run post-selection repeatedly, and then remove another level of the $[[7,1,3]]$ CSS code. Eventually, we remove all of the levels of the $[[7,1,3]]$ code and are left with a logical ancilla state in our desired code, which may be the $[[23,1,7]]$ code concatenated with itself repeatedly.

\section{Discussion of overhead}
\label{sec:converge}
In Sec. \ref{sec:entropy}, we found that the degeneracies of a code cause the errors on each qubit to cease being independent of each other, thereby lowering the threshold a bit. In this section, we will discuss how using finite resources affects the thresholds. We pick some noise rate below the threshold and a rate of error that we are willing to tolerate per logical gate, which results in a finite amount of overhead involved. We examine how the noise converges below the threshold.

\subsection{Post-selection overhead}
Suppose that there are $q$ logical qubits encoded in an $n$ qubit code, giving a total of $N=qn$ qubits. If the total bit flip noise on the destination qubits before measurement is $p$, then the probability of not rejecting everything is roughly $p_k \approx (1-p)^N$, and so we have to run this an average of $\frac{1}{p_k}$ times. The process described in Sec. \ref{encoding} results in an overhead of $o_s = \frac{1}{p_k} (o_{t-1} + o_{t-2}$ after $s$ post-selections. Then, roughly $o_s$ is $o( N (p_k^{-1}+1 - p_k)^s))$. After a while, the noise will tend toward an equilibrium point, getting a factor $r$ closer each time. Therefore, to get within $\epsilon$ of the equilibrium point requires $s$ to be of the order $\log_r \epsilon$. If $N$ is large, to get within $\epsilon$ of the equilibrium point will require something like $N (p_k^{-1})^{\log_r \epsilon} = N \epsilon^{-N log_r (1-p)}$ overhead. 

At the hashing bound for forward type noise, $p \approx 15.3$, so for Bell pairs for the $[[7,1,3]]$ code, the predicted probability of not rejecting is $p_k \approx (1-p)^{14} = 9.79\%$. The actual $p_k=9.21\%$. The observed rate of convergence is $r=0.766$. This results in $o(N\epsilon^{-0.665N})$ overhead. This is exponential overhead in the number of qubits, and doubly exponential in the number of levels of concatenation of the code.

\subsection{Convergence by level}

Suppose we have a distance $d$ code concatenated with itself $l$ times, the initial entropy was $t$, and the threshold entropy was $t_c$. Near the threshold, the entropy behaves roughly as $t \rightarrow t_c - (t_c - t) (d^l)^{\alpha}$, where $\alpha = \log_2 \log_2 e \approx 0.53$. 

Away from the threshold, if $\Delta = \frac{(t_c - t) d^{j \alpha}}{t_c}$, the rate of logical error has a magnitude of around $10^{-3}$ or $10^{-4}$ for $\Delta=2$, and goes to $0$ very rapidly after this. 

For example, suppose we have the $[[7,1,3]]$ CSS code concatenated with itself $l=4$ times and our noise is of the form $(0,0,p)$, and $p$ is $\frac{2}{3}$ of the threshold value of $10.963(2)\%$. Then $d=3$, $t_c = 0.49880$, and $t=0.3492$. Therefore, $\Delta \approx 2.5$. The probability of an encoded error is $2 \times 10^{-5}$. 

If the noise is some $\epsilon$ below the threshold, then to get the noise below a given value requires that $d^l$ be $o(\epsilon^{-\alpha^{-1}})$. If the code encodes $1$ qubit into $n$ qubits, then the number of qubits $n^l$ is $o(\epsilon^{-\frac{\log_d n}{\alpha}})$. For the $[[7,1,3]]$ code, this is $o(\epsilon^{3.35})$.

\ignore{
Suppose have a distance $d$ code concatenated with itself $n$ times, and the rate of errors is $p$, which is less than the threshold of $p_c$. At exactly the threshold, after a few levels of concatenation, the logical errors will have some constant rate of error ($1$ - fidelity) $p_o$, which is around $25\%$ to $30\%$.

The total distance of the code is $D=d^n$. In terms of $p$, the rate of error is proportional to $p^{D^{\alpha}}$, where $\alpha \approx 0.53$ and could be $\log_2 \log_2 e$. Then, the output logical error is
\begin{equation*}
p_l \approx (\frac{p}{p_c})^{d^{\alpha n}} p_o.
\end{equation*}

We can alternately write this as
\begin{equation*}
\frac{ \log\frac{p_l}{p_o}}{\log \frac{p}{p_c}} = D^{\alpha} = d^{\alpha n}.
\end{equation*}

Suppose we have a distance $3$ code and the rate of errors is half of that of the threshold ($p=0.5 p_c$). Suppose we wish to find how many levels of concatenation are required to have an output error of $10^{-6}$. Assuming  that $p_o = 30\%$, we solve
\begin{align*}
\frac{ \log \frac{10^{-6}}{0.3}}{\log 0.5} = 18.20\\
\sqrt[\alpha]{18.20} = 242 = 3^n < 3^5 = 243,
\end{align*}
and find that the 5th level of concatenation works.
}

\subsection{Fixed fidelity points}
Note that just below the threshold, the rate of an error increases for the first few levels of concatenation. Suppose we require that the rate of error be the same after the first level of the code as it was from the noise in Eq. \ref{eq:measure}. For the Knill type noise with the $[[7,1,3]]$ CSS code, $p=3.472\%$ has a fidelity of $p_I=0.90602$ in both cases. This can be thought of as a very crude proof that $p=3.472\%$ is not above the threshold. Other values are calculated in Tab. \ref{tab:fixedpoints}

\begin{table}
\caption{\label{tab:fixedpoints} Fidelity the same at 0th and 1st levels of code. \noises}
\begin{ruledtabular}
\begin{tabular}{cccc}
Code         & Noise        & p        &  fidelity  \\
\hline
$[[7,1,3]]$  & Knill        & 3.472\%  & 0.90602 \\
$[[7,1,3]]$  & Depolarizing & 4.039\%  & 0.91122 \\
$[[7,1,3]]$  & Forward      & 2.9595\% & 0.87703 \\
$[[23,1,7]]$ & Forward      & 3.5471\% & 0.85108 \\
\end{tabular}
\end{ruledtabular}
\end{table}

\section{Conclusion}
We find that for various CSS codes, particularly the $[[23,1,7]]$ CSS code, when extremely large amounts of resources (i.e. remotely interacting ancilla qubits) are available, the threshold is close to that of the hashing bound (when the Shannon entropy of the noise is $1$). We conjecture that the hashing bounds given in Table \ref{tab:hashingfault} represent an upper bound for fault tolerant thresholds. This bound is much lower than the best previously known upper bounds of $45\%$ \cite{BCLLSU} and  $29.3\%$ \cite{Kempe2008}. It is higher than a fault tolerant threshold for the adversarial noise model \cite{Kay}. 

The CSS codes are formed from classical codes and therefore do not exceed the Shannon hashing rating bound. The hashing bound had been conjectured to be the upper bound for non fault tolerant error correction, but it was found that codes can exceed this \cite{ShorSmolin,DiVincenzo98,SS,capacity}. However, the codes that exceed the hashing bound seem to do so because of their degeneracies, and these degeneracies substantially hurt the threshold, as is seen in Sec. \ref{sec:entropy}. The bit flip codes, which seem to yield the highest capacities \cite{capacity}, have $\binom{n}{2}$ ways of having $Z$ errors on two qubits, which results in no error. This effect on the threshold is large enough that the threshold will be below the hashing bound.

We conjecture that if {\tt CNOT} is the only multiple qubit gate implemented directly, the fault tolerant hashing bound described in this paper serves as an upper bound on the fault tolerant threshold. As can be seen from Tab. \ref{tab:2317thresholdvalues}, the $[[23,1,7]]$ CSS code provides high fault tolerant thresholds, which are at least $99.6\%$ of the hashing bound thresholds given in Tab. \ref{tab:hashingfault}, leaving little room for improvement.

In the future, we hope to apply these ideas to analyze the amount of overhead, and develop new realistic fault tolerant schemes. Unfortunately, it is unrealistic to get close to the fault tolerant thresholds given here. The resources needed to run quantum computing very close to these thresholds will be exceedingly high; to do so would require googols of universes with remotely interacting qubits. However, it would be useful to compare schemes with restricted amounts of overhead to the truly ultimate thresholds found in this paper. Since the thresholds found here are so high, achieving even some fraction of them with a reasonable amount of overhead may offer an improvement over current schemes. For example, we find a threshold of $6.9\%$ for Knill noise, which is much larger than the $1\%$ threshold with a reasonable amount of overhead found by Knill \cite{Knill}.

\begin{acknowledgments}
We thank K. B. Whaley for helpful discussions, and careful reading of the manuscript. Quantum circuits were drawn using Q-Circuit \cite{QCircuit}. 
\end{acknowledgments}

\appendix

\section{Code method}
\label{sec:codemethod}
In this section, we use the quantum error correction channel formalism of \cite{FKSS} to find the results of most of Sec. \ref{sec:basic} by an alternate means.

Suppose we have a code that encodes $k$ qubits into $n$ qubits. We have an operator $\mathcal{E}$ which encodes $4^k$ dimensional density vectors representing density states on $k$ qubits into $4^n$ dimensional density vectors (as shown in Eq. \ref{eq:vectordensity}) representing density states on $n$ qubits. The noise that occurs on these $n$ qubits is represented by a $4^n$ by $4^n$ dimensional operator $\mathcal{N}$. This is followed by measuring the stabilizer operators to determine the syndrome $i$. We next apply a recovery operator $R_i$ and decode to evaluate the resulting logical errors. This is the channel map technique described in  \cite{Rahn:02a,FKSS,MoreChans}.

The resulting noise map for each recovery operator $R_i$ is
\begin{equation*}
\mathcal{G}^{R_i} =  \frac{1}{2^{n-k}} \mathcal{E}^t \circ \mathcal{O}(R_i) \circ \mathcal{N} \circ \mathcal{E}.
\end{equation*}
The total noise map is $\mathcal{G} = \sum_i \mathcal{G}^{R_i}$.

Now we use these channel maps for a different way of looking at Sec. \ref{sec:basic}. If the physical data qubit was $\alpha \ket{0} + \beta \ket{1}$, we model this by having the $2$ physical qubits in the state $\alpha \ket{00} + \beta \ket{11}$, apply {\tt CNOT}, and then measure. $0$ represents no error and $1$ represents an error, just like in Fig. \ref{fig:PhysDet}. If this process is noise free, the data qubit is now $\alpha \ket{0} + \beta \ket{1}$.

\ignore{
\subsection{Calculations}
}
In Fig. \ref{fig:PhysDet}, we assumed that the source qubit has a pre-existing noise $\mathcal{S}$, and the destination qubit has a pre-existing noise $\mathcal{D}$. We then apply {\tt CNOT}, and the noise $\mathcal{Q}$ is created. We follow this by measuring the destination qubit.

We look at the two codes in Tab. \ref{tab:codes}. C1 is the two qubit bit flip code. By applying {\tt CNOT} from its first qubit to its second qubit, we change it into the C2 code, and vice versa.

\begin{table}
\caption{\label{tab:codes}Description of two codes}
\begin{ruledtabular}
\begin{tabular}{ccc}
              & C1         & C2 \cr
\hline
$\overline{\ket{0}}$     & $\ket{00}$  & $\ket{00}$ \cr
$\overline{\ket{1}}$     & $\ket{11}$  & $\ket{10}$ \cr
$S$                      & $\{II,ZZ\}$ & $\{II,IZ\}$ \cr
$\mathcal{E}_I$          & $II + ZZ$   & $II + IZ$ \cr
$\mathcal{E}_X$          & $XX + YY$   & $XI + XZ$ \cr
$\mathcal{E}_Y$          & $XY + YX$   & $YI + YZ$ \cr
$\mathcal{E}_Z$          & $ZI + IZ$   & $ZI + ZZ$ \cr
$R$                      & $\{II,IX\}$ & $\{II,IX\}$ \cr
\end{tabular}
\end{ruledtabular}
\end{table}

These can be written out in matrix form. For example for $C2$,

\begin{equation*}
  {\cal E}^t =
  \left[
  \begin{array}{cccccccccccccccc}
    1 & 0 & 0 & 1 & 0 & 0 & 0 & 0 & 0 & 0 & 0 & 0 & 0 & 0 & 0 & 0 \cr
    0 & 0 & 0 & 0 & 1 & 0 & 0 & 1 & 0 & 0 & 0 & 0 & 0 & 0 & 0 & 0 \cr
    0 & 0 & 0 & 0 & 0 & 0 & 0 & 0 & 1 & 0 & 0 & 1 & 0 & 0 & 0 & 0 \cr
    0 & 0 & 0 & 0 & 0 & 0 & 0 & 0 & 0 & 0 & 0 & 0 & 1 & 0 & 0 & 1
  \end{array}
  \right] .
\end{equation*}

Let $\mathcal{O}({\tt CNOT})$ be the superoperator for {\tt CNOT}. Together with its noise operator, the faulty {\tt CNOT} gate is $\mathcal{Q} \circ \mathcal{O}({\tt CNOT})$. Together with the pre-existing noise, we have
\begin{equation}
\mathcal{Q} \circ \mathcal{O}({\tt CNOT}) \circ \mathcal{S} \otimes \mathcal{D}.
\end{equation}
Note that this changes from code C1 to code C2. If we wish to do the calculations in terms of code C1, we apply {\tt CNOT} at the end. If we wish to do the calculations in terms of code C2, we apply {\tt CNOT} at the beginning. The logical noise from these codes are
\begin{align*}
\mathcal{N}_{C1} =  \mathcal{O}({\tt CNOT})  \circ \mathcal{Q} \circ \mathcal{O}({\tt CNOT}) \circ \mathcal{S} \otimes \mathcal{D} \\
\mathcal{N}_{C2} = \mathcal{Q} \circ \mathcal{O}({\tt CNOT}) \circ \mathcal{S} \otimes \mathcal{D} \circ \mathcal{O}({\tt CNOT}).
\end{align*}
By conjugating, {\tt CNOT} produces the following map $c$ which acts on two qubit Pauli matrices as $c(\sigma)$ as
\begin{align*}
(II), (IX), IY \leftrightarrow ZY, IZ \leftrightarrow ZZ, XI \leftrightarrow XX\\
 XY \leftrightarrow YZ, XZ \leftrightarrow YY, YI \leftrightarrow YX, (ZI), (ZX).
\end{align*}
This map has order $2$, that is, $c(c(\sigma)) = \sigma$. Note that bit flips ($X$ or $Y$ ) propagate forward as $X$, and phase flips ($Y$ or $Z$) propagate backwards as $Z$.

\subsection{Diagonal noise}
We now use the shorthand $\mathcal{N}_{\sigma} = \mathcal{N}_{\sigma, \sigma}$. The initial source noise has diagonal noise $\mathcal{S}=[1,S_X,S_Y,S_Z]$, and similarly the destination noise has the diagonal noise $\mathcal{D}=[1,D_X,D_Y,D_Z]$. Control can be used to make the noise diagonal, as was done with the Singular  Value Decomposition Theorem for channels \cite{FKSS}. This turns one qubit unital noise into diagonal noise, and any non-unital part has little effect on the threshold.

If there is no error detected, and $\sigma$ and $\sigma'$ are one qubit Pauli matrices, then
\begin{align*}
\mathcal{G}^{II}_{\sigma, \sigma'} = \frac{1}{2} (N_{\sigma I, \sigma' I} + N_{\sigma I, \sigma' Z} + N_{\sigma Z, \sigma' I} + N_{\sigma Z, \sigma' Z})\\
\mathcal{G}^{XI}_{\sigma, \sigma'} = \frac{1}{2} (N_{\sigma I, \sigma' I} + N_{\sigma I, \sigma' Z} - N_{\sigma Z, \sigma' I} - N_{\sigma Z, \sigma' Z}).
\end{align*}

In the case of diagonal noise, we have Eq. \ref{eq:Ntraceout}.

\section{Other gates and fault tolerant transversal gates}
\label{othergates}
For the doubly even (stabilizers have weights divisible by $4$) CSS codes discussed in this paper, the Clifford gates are encoded as themselves (or their adjoint) transversally \cite{Gottesman:97b}.   Measurement of a Pauli operator and {\tt CNOT} are also encoded as themselves transversally. Together with any $1$ qubit pure non Pauli eigenstate, these give universal quantum computation \cite{BWR}. 

Transversely encoded gates are very useful to prevent propagation of errors. For a more in depth discussion of how these are useful for fault tolerance, see \cite{Preskill}.

Typically, one prepares a "magic" state, $\frac{1}{\sqrt{2}}( \ket{0} + \sqrt{i} \ket{1})$ or $\cos(\frac{\pi}{8})\ket{0}+\sin(\frac{\pi}{8}) \ket{1}$ that can be used to create the $\frac{\pi}{8}$ gate.  We discuss $3$ approaches to get a "magic state":
\begin{itemize}
\item The method described by Nielsen and Chuang to measure the state using a cat state \cite{CN}.
\item There is a $15$ qubit Hamming CSS code, where the $\frac{\pi}{8}$ gate can essentially be encoded as itself \cite{SBAK}. If there is a Bell pair consisting of the $15$ qubit CSS code and some other code, this can be used to teleport the logical  qubit into another code. A similar method would be noting that the $15$ qubit code is similar to the $[[7,1,3]]$ CSS code, and converting it directly into that code.
\item We have a high distance code encoded in a Bell pair. One of the logical qubits is decoded into a physical qubit \cite{Knill3}, and then the $\frac{\pi}{8}$ gate is applied to that qubit, followed by a measurement in the $X$ eigenbasis. Then, the code is encoded into the $15$ qubit code, and post-selection is run on that as in \cite{SBAK}. 

\end{itemize}

\section{Noise}
\label{sec:noise}
In this section, we discuss our representation of noise, and derive a formula for the noise after post-selection.

If there is probability $p_x$ of an error, we write the noise as $x=1-2p_x$.

Since $x = (1-2p_y)(1-2p_z) = 1 - 2 (p_y + p_z - 2p_y p_Z)$, two noises $y$ and $z$ are combined as $x=yz$.

Suppose a qubit with noise $x_1$ has {\tt CNOT} applied to a qubit with noise $x_2$, and post-selection is performed by measuring the 2nd qubit in the standard basis. We either have no error with probability $p_g = (1-p_1)(1-p_2)$, or an undetected error with probability $p_b = p_1 p_2$. The new probability of an error is then $p = \frac{p_b}{p_g+p_b}$, and so
\begin{align*}
\text{post }(x_1, x_2) = 1 -2p =  \frac{p_g-p_b}{p_g+p_b}\\
 = \frac{2 - 2p_1 - 2p_2}{2 -2p_1 -2p_2 + 4p_1p_2} = \frac{x_1 + x_2}{1+x_1 x_2}.
\end{align*}

\section{$[[7,1,3]]$ CSS code}
\label{sec:713}
In this section, we find non fault tolerant probabilities of errors for the $[[7,1,3]]$ CSS code, which are applied earlier for fault tolerant calculations. 

The $[[7,1,3]]$ CSS code is generated from the classical $[7,4,3]$ Hamming code given by the parity check matrix
\begin{equation}
\label{matrixp}
\begin{bmatrix}
0 & 0 & 0 & 1 & 1 & 1 & 1\cr
0 & 1 & 1 & 0 & 0 & 1 & 1\cr
1 & 0 & 1 & 0 & 1 & 0 & 1\cr
\end{bmatrix} .
\end{equation}

The stabilizer generators for the $[[7,1,3]]$ CSS code are
\begin{align*}
g_1= IIIXXXX &&
g_4= IIIZZZZ \\
g_2 = IXXIIXX &&
g_5 = IZZIIZZ \\
g_3 = XIXIXIX &&
g_6 = ZIZIZIZ
\end{align*}

The encoded Pauli operators are
\begin{equation*}
\overline{X} = X^{\otimes 7}, \overline{Z} = Z^{\otimes 7}, \overline{Y} = i \overline{X} \overline{Z} = - Y^{\otimes 7}
\end{equation*}
A circuit for generating the encoded $\overline{\ket{0}}$ is given in Fig. \ref{fig:encode}.

\subsection{Probabilities of errors of a given distance}
Suppose we have the $[[7,1,3]]$ CSS code, with uncorrelated bit flip and phase flip noise. Then the two types of noise can be corrected independently of each other. If we wish to find the probability $m_i$ of a distance $i$ phase flip or bit flip error, we need just consider the $[7,4,3]$ classical code that generates the CSS code. 

Looking at the parity check in Eq. \ref{matrixp}, we see that there are $8$ different possible error syndromes. The $0$ syndrome corresponds a distance $0$ error (no error) or distance $3$ error (undetected error). The other syndromes correspond to either a distance $1$ or distance $2$ error. The code corrects any one error. Counting the $2^7$ different possible errors, and their distances, we get Tab. \ref{tab:713dist}. 

\ignore{
\begin{equation}
\label{eq:713dist}
\begin{array}{c|c|c|c|c|c}
\text{Num. errors} & d=0 & d=1 & d=2 & d=3 & \text{Tot. cases}\cr
\hline
0 & 1 & 0 & 0 & 0 & 1 \cr
1 & 0 & 7 & 0 & 0 & 7 \cr
2 & 0 & 0 & 21 & 0 & 21 \cr
3 & 0 & 28 & 0 & 7 & 35 \cr
4 & 7 & 0 & 28 & 0 & 35 \cr
5 & 0 & 21 & 0 & 0 & 21 \cr
6 & 0 & 0 & 7 & 0 & 7 \cr
7 & 0 & 0 & 0 & 1 & 1 \cr
\end{array}.
\end{equation}
}

\begin{table}
\caption{\label{tab:713dist} Distances $d$ of errors of weight $w$ (errors on $w$ qubits) for the $[7,4,3]$ classical code that generates the $[[7,1,3]]$ CSS code.}
\begin{ruledtabular}
\begin{tabular}{cccccc}
$w$ & $d=0$ & $d=1$ & $d=2$ & $d=3$ & \text{Tot. cases} \\
\hline
0 & 1 & 0 & 0 & 0 & 1 \\
1 & 0 & 7 & 0 & 0 & 7 \\
2 & 0 & 0 & 21 & 0 & 21 \\
3 & 0 & 28 & 0 & 7 & 35 \\
4 & 7 & 0 & 28 & 0 & 35 \\
5 & 0 & 21 & 0 & 0 & 21 \\
6 & 0 & 0 & 7 & 0 & 7 \\
7 & 0 & 0 & 0 & 1 & 1 \\
\end{tabular}
\end{ruledtabular}
\end{table}

In terms of the quantum stabilizer code, distance $0$ errors (no error) correspond to members of the stabilizer group $S$, which preserve the codespace. Distance $3$ errors are encoded errors that commute with the stabilizer elements, and so are in $C(S) \setminus S$, where the $C$ represents centralizer. Distance $1$ and $2$ errors are errors outside of the codespace.

If we have independent probability $p$ of an error on each qubit, then the probability of a distance $d$ error is $m_d$, where
\begin{align*}
m_0 =(1-p)^7+7p^4(1-p)^3\\
m_1 =7 p (1-p)^6 + 28 p^3 (1-p)^4 + 21 p^5 (1-p)^2 \\
m_2 =21 p^2 (1-p)^5+28p^4(1-p)^3 + 7 p^6 (1-p)\\
m_3 = 7 p^3 (1-p)^4 + p^7.
\end{align*}
If we use $x=1-2p$ instead, we get

\begin{equation}
\label{eq:m_i}
\begin{bmatrix}
m_0 \cr
m_1\cr
m_2\cr
m_3\cr
\end{bmatrix} = \frac{1}{16} \begin{bmatrix}
1 & 7 & 7 & 1 \cr
7 & 7 & -7 & -7 \cr
7 & -7 & -7 & 7 \cr
1 & -7 & 7 & -1 \cr
\end{bmatrix} \begin{bmatrix}
1 \cr
x^3 \cr
x^4 \cr
x^7 \cr
\end{bmatrix}.
\end{equation}

Alternately, we see from the centralizer that $m_0=\frac{1+7x^3+7x^4+x^7}{16}$. This represents a distance $0$ error (no error).  What happens to each of the terms as we add an error of a fixed $d$? We find all the ways to add an error, and replace $x$ with $-x$ on that particular qubit, and then average the result. We get
\begin{equation}
\begin{array}{c|c|c|c}

d=0 & d=1 & d=2 & d=3 \cr
1 & 1 & 1 & 1 \cr
x^3 & \frac{1}{7}x^3 & -\frac{1}{7}x^3 & -x^3 \cr
x^4 & -\frac{1}{7}x^4 & -\frac{1}{7}x^4 & x^4 \cr
x^7 & -x^7 & x^7 & -x^7 \cr
\end{array}.
\end{equation}

 For this code, an encoded error is a distance $2$ or $3$ error. Since the $0$ syndrome is either distance $0$ or $3$, and the other $7$ error syndromes are each either distance $1$ or distance $2$, and the syndrome calculation is linear over $F_2^7$,  we get $16$ equivalence classes of $8$ errors each. We can choose a representative element of each, for example all of the $0$, $1$, $6$, and $7$ qubit $Z$ errors. We then calculate what happens when errors are combined by finding the distance of the combined error. For example, if we pick $2$ random distance $1$ errors, we have a $\frac{1}{7}$ probability of a distance $0$ error, and a $\frac{1}{6}$ probability of a distance $2$ error. Then, if we combine a probability $a_d$ of having distance $d$ errors with a probability  $b_d$ of having distance $d$ errors, we get a combined probability $a'_d$ of having distance $d$ errors

\begin{equation}
\label{newerrtype}
\begin{bmatrix}
a'_0 \cr
a'_1 \cr
a'_2 \cr
a'_3 \cr
\end{bmatrix} =
\begin{bmatrix}
b_0 & \frac{1}{7} b_1 & \frac{1}{7} b_2 & b_3 \cr
b_1 & b_0 + \frac{6}{7} b_2 & b_3 + \frac{6}{7} b_1 & b_2 \cr
b_2 & b_3 + \frac{6}{7} b_1 & b_0 + \frac{6}{7} b_2 & b_1 \cr
b_3 & \frac{1}{7}b_2 & \frac{1}{7} b_1 & b_0 \cr
\end{bmatrix}
\begin{bmatrix}
a_0 \cr
a_1 \cr
a_2 \cr
a_3 \cr
\end{bmatrix} .
\end{equation}

If we post-select on $a$ and $b$, then we keep the result with probability
\begin{equation*}
p_k = a_0 b_0 + \frac{1}{7} a_1 b_1 + \frac{1}{7} a_2 b_2 + a_3 b_3,
\end{equation*}
and the result is
\begin{equation}
\frac{1}{p_k} \begin{bmatrix}
a_0 b_0 \cr
\frac{1}{7} a_1 b_1 \cr
\frac{1}{7} a_2 b_2 \cr
a_3 b_3 \cr
\end{bmatrix}.
\end{equation}

\end{document}